\documentclass{sig-alternate}
\usepackage{ifthen}
\newboolean{short} 
\setboolean{short}{false}
\begin{document}

\newtheorem{lemma}{Lemma}
\newtheorem{theorem}[lemma]{Theorem}
\newtheorem{informaltheorem}[lemma]{Informal Theorem}
\newtheorem{informallemma}[lemma]{Informal Lemma}
\newtheorem{corollary}[lemma]{Corollary}
\newtheorem{definition}[lemma]{Definition}
\newtheorem{proposition}[lemma]{Proposition}
\newtheorem{question}{Question}
\newtheorem{problem}{Problem}
\newtheorem{remark}[lemma]{Remark}
\newtheorem{claim}{Claim}
\newtheorem{fact}{Fact}
\newtheorem{challenge}{Challenge}
\newtheorem{observation}{Observation}
\newtheorem{openproblem}{Open Problem}
\newtheorem{openquestion}{Open question}
\newenvironment{LabeledProof}[1]{\noindent{\bf Proof of #1: }}{\qed}

\newcommand{\Gtensor}[0]{G \times G} 
\newcommand{\D}[0]{\mathcal{D}} 
\newenvironment{proofsketch}{\trivlist\item[]\emph{Proof Sketch}:}%
{\unskip\nobreak\hskip 1em plus 1fil\nobreak$\Box$
\parfillskip=0pt%
\endtrivlist}

\newcommand{\beq}{\begin{equation}}
\newcommand{\eeq}{\end{equation}}
\newcommand{\beas}{\begin{eqnarray*}}
\newcommand{\eeas}{\end{eqnarray*}}

\newcommand{\poly}{\mathrm{poly}}
\newcommand{\eps}{\epsilon}
\newcommand{\e}{\epsilon}
\newcommand{\E}{\mathbb{E}}
\newcommand{\polylog}{\mathrm{polylog}}
\newcommand{\rob}[1]{\left( #1 \right)} 
\newcommand{\sqb}[1]{\left[ #1 \right]} 
\newcommand{\cub}[1]{\left\{ #1 \right\} } 
\newcommand{\rb}[1]{\left( #1 \right)} 
\newcommand{\abs}[1]{\left| #1 \right|} 
\newcommand{\zo}{\{0, 1\}}
\newcommand{\zonzo}{\zo^n \to \zo}
\newcommand{\zokzo}{\zo^k \to \zo}
\newcommand{\zot}{\{0,1,2\}}
\newcommand{\stirlingtwo}[2]{\genfrac{\{}{\}}{0pt}{}{#1}{#2}}
\newcommand{\en}[1]{\marginpar{\textbf{#1}}}
\newcommand{\efn}[1]{\footnote{\textbf{#1}}}

\newcommand{\junk}[1]{}
\newcommand{\prob}[1]{\Pr\left[ #1 \right]}
\newcommand{\expt}[1]{\mbox{E}\sqb{#1}}
\newcommand{\expect}[1]{\mbox{E}\left[ #1 \right]}

\newcommand{\md}[1]{\delta_{#1}}                            
\newcommand{\gf}[1]{G_{#1}}                                     
\newcommand{\nei}[3]{N^{#2}_{#1}\rb{#3}}                
\newcommand{\dg}[2]{d_{#1}\rb{#2}}                          
\newcommand{\dgi}[3]{d_{#1}\rb{#2,#3}}                   
\newcommand{\prb}[2]{p_{#1,#2}}

\newcommand{\BfPara}[1]{\noindent {\bf #1}.}

\newcommand{\shortOnly}[1]{\ifthenelse{\boolean{short}}{#1}{}}
\newcommand{\onlyShort}[1]{\ifthenelse{\boolean{short}}{#1}{}}
\newcommand{\longOnly}[1]{\ifthenelse{\boolean{short}}{}{#1}}
\newcommand{\onlyLong}[1]{\ifthenelse{\boolean{short}}{}{#1}}
\newboolean{long}
\setboolean{long}{true}
\clubpenalty=10000
\widowpenalty = 10000

\title{Better bounds for coalescing-branching random walks}
\numberofauthors{3} 
%
\author{
%
%
\alignauthor
Michael Mitzenmacher \titlenote{
This work was
supported in part by NSF grants CNS-1228598, CCF-1320231, and CCF-1535795.} \\
       \affaddr{Harvard University}\\
       \email{michaelm@eecs.harvard.edu}
\alignauthor Rajmohan Rajaraman\titlenote{
       Supported in part by NSF CCF-1422715, NSF CCF-1535929, and an
       ONR grant} \affaddr{Northeastern
  University}\\ \email{rraj@ccs.neu.edu}
\and \alignauthor Scott Roche \titlenote{ Supported in part by NSF
  CCF-1216038 and NSF CCF-1422715.} \\ \affaddr{Northeastern
  University and Akamai Technologies}\\ \email{str7@cornell.edu} }

\maketitle
\begin{abstract}
Coalescing-branching random walks, or {\em cobra walks} for short, are
a natural variant of random walks on graphs that can model the spread
of disease through contacts or the spread of information in networks.
In a $k$-cobra walk, at each time step a subset of the vertices are
active; each active vertex chooses $k$ random neighbors (sampled
indpendently and uniformly with replacement) that become active at the
next step, and these are the only active vertices at the next step.  A
natural quantity to study for cobra walks is the cover time, which
corresponds to the expected time when all nodes have become infected
or received the disseminated information.  \junk{Suprisingly little is
  known about the cover time for cobra walks.}

In this work, we extend previous results for cobra walks in multiple
ways.  We show that the cover time for the 2-cobra walk on $[0,n]^d$
is $O(n)$ (where the order notation hides constant factors that depend
on $d$); previous work had shown the cover time was $O(n \cdot \polylog(n))$.
We show that the cover time for a 2-cobra walk on an $n$-vertex
$d$-regular graph with conductance $\phi_G$ is $O(\phi_G^{-2} \log^2
n)$, significantly generalizing a previous result that held only for
expander graphs with sufficiently high expansion.  And finally we show
that the cover time for a 2-cobra walk on a graph with $n$ vertices is
always $O(n^{11/4} \log n)$; this is the first result showing that the
bound of $\Theta(n^3)$ for the worst-case cover time for random walks
can be beaten using 2-cobra walks.
\end{abstract}




\section{Introduction}
Random walks provide a fundamental mathematical model for many basic
network processes.  In disease models, transmission of a virus can be
modeled by the virus moving according to a random walk on a graph
representing a human contact network; computer viruses can be modeled
similarly~\cite{GANESH,PIET,Dutta:2015:CRW:2821462.2817830}.  Variants of random
walks can also be used for information dissemination, using
message-passing protocols where a message is passed from neighbor to
neighbor via a random walk~\cite{sicomp,feige-rumor}.  Such protocols
require little state information and are robust to various types of
faults, and are therefore useful in many distributed networks~\cite{7354403}.   More
generally, random walks provide a fundamental primitive for network
algorithms for information propagation, search, routing, and load
balancing.

In many of these settings, a central measure of interest is the cover
time, the expected time for the random walk to cover all of the
vertices of the underlying network.  In disease models, this
corresponds to the time until all vertices in the network have been
exposed to the virus; in message-passing protocols, this corresponds
to the time until all vertices have received the message.

Parallel random walks provide a natural generalization of standard
random walks, with multiple random walks traversing the network
simultaneously, and several papers have analyzed the performance of
parallel random walks (as we describe in the related work section).  A
related variant, less well-studied and understood, are
coalescing-branching random walks, or {\em cobra walks} for
short~\cite{Dutta:2015:CRW:2821462.2817830}.  In a cobra walk, at each
time step, a subset of the vertices are active; typically in the
initial state a single vertex would be active.  At each time step,
each active vertex chooses $k$ random neighbors (sampled independently
and uniformly with replacement) that become active at the next step.
A vertex is active at step $t$ if and only if it was chosen by an
active vertex in the previous step.  When $k > 1$, each walk branches
at that step into multiple walks, but multiple walks then coalesce
when they reach the same vertex at the same time.  We refer to a cobra
walk where at each time step a vertex chooses $k$ active neighbors as
a $k$-cobra walk for convenience.  (One could further study variations
where the branching varied based on the vertex or the time step, or
was governed by a random distribution; we do not do that here.)

As examples of cobra walks, in the message passing setting, a
$k$-cobra walk corresponds to a network where a vertex may send $k$
outgoing copies of the message to neighbors during a time step instead
of just one.  In disease networks, a cobra walk corresponds to an
idealized process within the Susceptible Infected Susceptible model
(or SIS model): in each time step, an infected agent infects $k$
random neighbors and recovers, but can be infected again (including at
the next time step).

In contrast to results in parallel random walks, where the number of
walks is a parameter, in cobra walks the number of active vertices
varies over time and its behavior depends significantly on the
network.  One might expect cobra walks to yield significant
improvements in the cover time over standard random walks, based on
their power to reproduce, even if limited by coalescence.  The goal of
our work is to formally and theoretically bound the performance of
cobra walks, focusing on the cover time.  While our results have
potential applications to distributed protocols and disease models, as
suggested above, we also believe that $k$-cobra walks are a natural
mathematical model worthy of study in their own right.

\subsection{Our Results and Techniques}
We are motivated by the prior work
\cite{Dutta:2015:CRW:2821462.2817830}, which obtained bounds on the
cover time of cobra walks on trees, grids, and expanders.  Our work
pushes those results further, in several directions.  Our primary
results are the following:
\begin{itemize}
\item We show that the cover time for the 2-cobra walk on $[0,n]^d$ is
  $O(n)$, where the constant in the order notation can depend on $d$.
  This improves on the previous bound of $O(n\cdot \polylog(n))$
  \cite{Dutta:2015:CRW:2821462.2817830}.  With respect to $n$, our
  result is optimal.
\item We show that the cover time for the 2-cobra walk for a
  $d$-regular graph with conductance $\phi_G$ is $O(\phi_G^{-2} \log^2
  n)$.  This generalizes a similar result in
  \cite{Dutta:2015:CRW:2821462.2817830} for expander graphs with
  sufficiently high expansion.  Our new result holds for any
  $d$-regular graph, and expresses the bound as a function of the
  conductance.
\item We provide a result for general graphs, showing that the cover
  time for a 2-cobra walk is always $O(n^{11/4} \log n)$.  For
  standard random walks, there are graphs where the cover time is
  $\Theta(n^3)$.  This is the first result showing that cobra walks
  can beat the corresponding worst-case bound for random walks.  We
  also establish an $O(n^{2-1/d}\log n)$ upper bound on the cover time
  for arbitrary $d$-regular graphs, again improving the tight
  quadratic bound for standard random walks.
\end{itemize}

Our main techniques involve making use of the parallelism inherent in
cobra walks, and by thinking of cobra walks as a {\em union of biased
  random walks}.  In some settings we can show the cobra walk goes
through an initial phase that instantiates a large number of
essentially parallel random walks, and then analyze the behavior of
these random walks.  Here we have to take care of the dependency
challenges introduced by coalescing, as random walks can essentially
disappear when several collide at a vertex.  In other settings, we
think of our cobra walk as being a single walk moving toward a
specific vertex, and then eventually taking a union bound over all
vertices.  At each step, we can choose to follow the active vertex
that moves toward the target vertex and discard the others.  This
approach simplifies the analysis by allowing us to focus on a single
walk, where now the $k$ choices correspond to a bias in the walk that
we can model.  The downside is such an analysis, however, is that it
does not take full advantage of the power of parallelism inherent to
cobra walks.

\subsection{Background and Related Work}

The cobra walk process has structural similarities with several other fairly diverse stochastic processes: branching processes, gossip protocols, and random walks (including parallel random walks, coalescing random walks, and other variants).  Despite these commonalities, cobra walks resist being fully described by any of these other processes; furthermore, analysis techniques used for these other processes often have no clear use or power in analyzing cobra walks.

\noindent{\bf Branching and coalescing processes}. Branching processes appear in many disciplines, from nuclear physics to population genetics, on various discrete and continuous structures~\cite{MR0163361,Madras1992255,benjamini2010trace}.  Another related topic is the study of coalescing processes and voter models (see, for example~\cite{cooper2012coalescing}). Naturally, there is also work on processes that contain both branching and coalescing elements~\cite{arthreya2005branching,sun2008brownian}, although unlike our work these analyses tend to operate in continuous time and on either restricted topologies or infinite spaces. The restriction of cobra walks to discrete time effectively disallows the use of differential-equation based analysis that often yields results in continuous-time processes (see, for example~\cite{GANESH,KES,PIET}). 

\noindent{\bf Gossip and rumor-spreading mechanisms} Gossip-based
algorithms have been used successfully to design efficient distributed
algorithms for a variety of problems in networks such as information
dissemination, aggregate computation, constructing overlay topologies,
and database synchronization (e.g., see \cite{sicomp} and the
references therein). There are three major variants of gossip-based
processes: push-based models, in which the members of the set of
informed vertices each select a neighbor and inform that neighbor (if
it is not already informed), pull-based models, in which uninformed
vertices select neighbors and poll them for information, and
push-pull, which is a combination of the first two.

Cobra walks bear the closest resemblance to push-based gossip
models. Indeed,~\cite{feige-rumor} show that the push process
completes in every undirected graph in $O(n \log n)$ steps with high
probability, and this bound has been conjectured to hold for cobra
walks \cite{Dutta:2015:CRW:2821462.2817830}. However, the similarity
between the two is in many ways superficial. If we view a gossip
process as a Markov chain on the state space of the set of all subsets
of vertices (representing the sets of possible informed vertices),
this Markov chain has a single absorbing state (assuming the graph is
connected) in which every vertex is informed. On the other hand,
performing a similar projection of a cobra walk onto a Markov chain of
the $2^n$ possible subsets of vertices that could be active at any
time, we see that there is no absorbing state and with the addition of
self-loops, the chain can be made ergodic.  

\noindent{\bf Random walks and parallel random walks.}  Cobra walks
also resemble standard random walks and parallel variants. For simple
random walks, the now classic work of Feige \cite{feige1,feige2}
showed that the cover time on any graph lies between $\Theta(\log n)$
and $O(n^3)$.  A formal model of biased random walks was introduced
in~\cite{azar} with the motivation of studying imperfect sources of
randomness.  Specifically, these biased walks allow a controller to
fix, at each step of the walk, the next step with a small probability,
with the aim of increasing the stationary probability at a target set
of vertices.  A variant of the biased walk of~\cite{azar} plays a
significant role in our analysis of cobra walks for general graphs.

Additional work has considered speeding up the cover time by modifying
the underlying process.  Adler et al~\cite{AHKV03} studied a process
on the hypercube in which in each round a vertex is chosen uniformly
at random and covered; if the chosen vertex was already covered, then
an uncovered neighbor of the vertex is chosen uniformly at random and
covered.  For any $d$-regular graph, Dimitrov and Plaxton showed that
a similar process achieves a cover time of $O(n + (n \log
n)/d)$~\cite{DP05}.  For expander graphs, Berenbrink et al\ showed a
simple variant of the standard random walk that achieves a linear
(i.e., $O(n)$) cover time~\cite{berenbrink}.

Parallel random walks, first studied in~\cite{broder} for the special
case where the starting vertices are drawn from the stationary
distribution and in~\cite{AAKKLT} for arbitrary starting vertices,
also appear related to cobra walks. 
Nearly-tight results on the speedup of cover time as a function of the
number of parallel walks have been obtained by~\cite{ElsasserS09} for
several graph classes including the cycle, $d$-dimensional meshes,
hypercube, and expanders. 
However, again the similarity is somewhat 
superficial. A parallel random walk with $k$ independent walks can be mapped
to a undirected random walk on a graph known as the tensor product. 
As such, much of the machinery of the analysis of simple random walks can be
applied to the parallel case. Applying a similar approach to a cobra walk is not
feasible, although for cobra walks one can convert the tensor product into a directed graph,
changing the topology significantly.  As such, generally the techniques
that can be used for parallel random walks cannot be used directly for the cobra walk. Indeed, one can view the dependencies on the positions of the other pebbles in a cobra walk as a manifestation of this difference. Cobra walks suffer from
the ``time's arrow'' effect: locally, most individual steps are
reversible, but as a cobra walk expands, the likelihood that it will
coalesce back to a single vertex grows exponentially unlikely.
Despite these difficulties, the tensor product graph can be useful when studying the movement of a small number of pebbles in a cobra walk, and we make use of this technique in obtaining a general 
bound based on conductance.

\section{Preliminaries} 

Let $G$ be a connected graph with vertex set $V$ and edge set $E$, and let $|V|=n$, except
for the case when we are analyzing the grid, in which case we let $V = [0,n]^d$. A $k$-coalescing-branching ($k$-cobra)
walk is defined as follows: It starts at time $t=0$ at an arbitrary vertex $v$, at which a pebble is
placed. In the next and every subsequent time step, every pebble in $G$ clones itself $k-1$ times
(so that there are now $k$ indistinguishable pebbles at each vertex that originally had a pebble). 
Each pebble then independently selects a neighbor of its current vertex uniformly at random
and moves to it. Once all pebbles have made their moves, the coalescing phase begins: if 
two or more pebbles are at the same vertex they coalesce into a single pebble, and the next round begins.

For time step $t$, $S_t$ is the \textbf{active set}, which is the set
of all vertices of $G$ that have a pebble. Define the \textbf{cover
  time} of a cobra walk to be the maximum over all vertices $v$ of
expectation of the minimum time $T$ at which all vertices have
belonged to some $S_t$ for $t \leq T$ when the cobra walk is started
at $v$.  We note that while our results are stated as bounds on the
cover time, all of the results in this paper actually give bounds on
the time to cover all the vertices in the graph with high probability,
as is clear from the proofs.  Hence we may also refer to the time at which
all vertices have been covered, where the meaning is clear.
The \textbf{hitting time} $H(u,v)$ is
the expectation of the minimum time it takes for any pebble
originating from a cobra walk that starts at $u$ to reach $v$.  The
\textbf{maximum hitting time} $h_{\max}$ is $max_{u,v \in V} H(u,v)$.

We make use of an extension of Matthews' Theorem, which relates the cover time
of a random walk to the maximum hitting time. The following theorem was 
proven in~\cite{Dutta:2015:CRW:2821462.2817830}:
\begin{theorem}
\label{matthewscobra}
Let $G$ be a connected graph on $n$ vertices. Let $W$ be a cobra walk on $G$ starting at
an arbitrary vertex. Then the cover time of $W$ on $G$ is bounded above by 
$O(h_{\max} \log n)$;  in fact $W$ covers all of $G$ in $O(h_{\max} \log n)$ steps
with high probability.  
\end{theorem}

Finally, we make use of a combinatorial property of the graph, the
conductance. Define the conductance of a set $S \subseteq V$ as
$\phi(S) = |\partial(S)|/vol(S)$, where $\partial(S) = \sum_{(u,v): u
  \in S, v \notin S} 1$ and $vol(S) = \sum_{u \in S} d(u)$. Then the
conductance $\Phi_G$ of the graph is $\min_{S : vol(S) \leq vol(V)/2}
\phi(S)$.  For the purposes of this paper, we say that a $d$-regular
graph is an $\epsilon$-expander if the conductance of the graph is
greater than or equal to $\epsilon$.

\section{Tight results for grids} 

We show that the cover time for the $d$-dimensional grid using a
2-cobra walk $[0,n]^d$ is $O(n)$, where the order notation hides
constant factors and other terms that depend on $d$; indeed, we show
all vertices are covered in $O(n)$ steps with high probability.
Previous work has shown that the cover time is $O(n \polylog(n))$
\cite{Dutta:2015:CRW:2821462.2817830}.\footnote{We note that the
  results of \cite{Dutta:2015:CRW:2821462.2817830} use a slightly
  different notation,
working with $n$ total nodes, or $[0,n^{1/d}-1]^d$.  We have opted
to work over $[0,n]^d$ for convenience.}  Our result
is clearly tight in its dependence on $n$.  Moreover, it shows that in
some circumstances one can avoid using tools such as Matthews' Theorem, 
which had been used previously in this setting \cite{Dutta:2015:CRW:2821462.2817830}, and
necessarily adds in an additional logarithmic factor in the number of
vertices over the hitting time.

The case of $d=2$ is simple and instructive;  we sketch a proof, but do
not go into full detail as we have a more detailed proof for the general case.
\begin{lemma}
\label{lem:22d}
The 2-cobra walk on $[0,n]^2$ has cover time $O(n)$.  
\end{lemma}
\begin{proof}
Let $v_0$ be the starting vertex of the walk.  We show that with high
probability all vertices are reached within $O(n)$ steps, and the
result follows, because if every vertex is hit within $T$ steps with
probability $p$, the expected time to cover all vertices is bounded
above by $T/p$.  

Let $v_1 = (x_1,y_1)$ be some other vertex on the grid.  Let $X_t$ be
the Manhattan distance between the closest pebble of the cobra walk
and $v_1$ after $t$ steps, which we will refer to as time $t$. (All
distances in this section will refer to Manhattan distances.)  Let
$u_t = (a_t,b_t)$ be some arbitrary vertex with a pebble of distance
$X_t$ from $v_1$ at time $t$.  We show by cases that there is drift so
that the expectation of $X_t$ decreases linearly over time, even when
at each step we pessimistically consider only the single vertex $u_t$
and not other additional pebbles; it follows from standard results in
random walks $v_1$ is reached after $O(n)$ steps with probability
$1-O(1/n^3)$.  Hence all vertices are covered after $O(n)$ steps with
probability $1-O(1/n)$, and the result follows.

If $x_1 \neq a_t$ and $y_1 \neq b_t$, the probability that at
least one of the two pebbles at $u_t$ moves closer to $v_1$ is at least
$1-(1/2)^2 = 3/4$, since each pebble moves closer to $v_1$ with
probability at least $1/2$.  (It can be more if a pebble is $v_1$ is on
the boundary of a grid.)  If either $x_1 = a_t$ or $y_1 = b_t$ but not
both, the probability that at least one of the two pebbles at $u_t$
moves closer to $v_1$ could be as small as $1-(3/4)^2 = 7/16 < 1/2$,
since each pebble could moves close to $v_1$ with probability only
$1/4$.  (This probability would be $1-(2/3)^2 = 5/9$ if for example
$x_1 = a_t = n$, but off the grid boundary this is not the case.)
Over one step, then, the expected distance may be increasing.  So we
instead consider two steps.  It is important to note that if the
distance increases on the first of the two steps but (at least) one
pebble moves so that $x_1 \neq a_t$ and $y_1 \neq b_t$ after the first
step, it improves our probability that the distance decreases in the
second step.  We find taking cases, assuming
that $u_t$ is at least distance 2 from the boundary and from $v_1$, that
$X_t$ increases by 2 by with probability $\frac{1}{16}\frac{9}{16}+\frac{1}{2}\frac{1}{4} = \frac{41}{256}$.
$X_t$ decreases by 2 by with probability $\frac{7}{16}\frac{7}{16} = \frac{49}{256}$.  
(Similar (better) results can be shown when $u_t$ is near a boundary.)

We therefore see that $X_t$ has negative drift (except at $X_t = 1$,
where the drift is slightly positive), and therefore the time to reach
0 can be shown to be $O(n)$ with probability $1-O(1/n)$ as claimed.
\end{proof}

The analysis above suggests technical difficulties to overcome with a direct
approach for general $d$; the behavior is slightly different at the boundaries of the
grid (though one could always work on the toroidal grid), and when 
coordinates match in one or more dimensions, it makes analyzing the drift
more difficult.  The analysis makes clear that for any fixed $d$, there should
be a large enough constant value of $k$ so that the cover time for the $k$-cobra walk 
is $O(n)$ on the $d$-dimensional grid $[0,n]^d$.   One just needs a large enough value
of $k$ so that the pebble nearest a target vertex drifts toward that vertex. 

We actually show the stronger result that the 2-cobra walk has cover
time $O(n)$ on the $d$-dimensional grid $[0,n]^d$, where the order
notation hides constants that depend on $d$.  We prove this below; we
have not aimed to optimize the constant factors.  The intuition for
the proof is the following.  If we look at the distance between the
closest point on the cobra walk and our target vertex in any single
dimension, it behaves like a biased random walk, with a bias toward 0.
Hence, after $O(n)$ steps, each individual dimension has matched
coordinates, with high probability.  Indeed, if each chain was an
independent biased random walk with constant bias (independent of
$d$), after $O(dn)$ steps we would expect each independent walk to be
near its stationary distribution, in which case each chain would be at
$0$ with some constant probability $\gamma$, and hence it would take
roughly $\gamma^d$ steps for all chains to be at 0 simultaneously.
Sadly, as usual, the fact that the chains are not themselves
independent causes significant technical challenges.

\begin{theorem}
\label{thm:dgrid}
The 2-cobra walk has cover time $O(n)$ on the $d$-dimensional grid
$[0,n]^d$ for any constant $d$. 
\end{theorem}
\begin{proof}
We break the proof into steps.  As in Lemma~\ref{lem:22d}, we consider
the distance to a target vertex over all dimensions for some pebble
generated by the 2-cobra walk.  We pessimistically keep track of only
a single pebble.  Specifically, our state at time $t$ can be defined as
follows.  Let $(z_{1,0},z_{2,0},\ldots,z_{d,0})$ be such that the
distance from the initial pebble to the target vertex is $z_{i,0}$ for
the $i$th grid dimension.  More generally, assume there is some pebble
at the $t$th step so that $(z_{1,t},z_{2,t},\ldots,z_{d,t})$ gives the
distance from that pebble to the target vertex in each dimension.  We
update the $z_{i,t}$ values over time steps as follows.  If our two
choices of pebbles generated from that pebble move in the same
dimension, we choose the pebble that moves closer to the target, if
such a pebble exists.  If our two choices of pebbles generated from
that pebble move in different dimensions $i$ and $j$, there are several
cases.  If $z_{i,t} = 0$ but $z_{j,t} \neq 0$, we choose the pebble
that moves in dimension $j$.  If $z_{i,t} = 0$ and $z_{j,t} = 0$, we
choose the pebble randomly.  If $z_{i,t} \neq 0$ and $z_{j,t} \neq 0$,
if both choices of pebbles move closer or both move farther away from
the target, we choose the pebble randomly; otherwise we choose the
pebble that moves closer.
\begin{lemma}
\label{lem:step1}
In each dimension, we have that if $z_{i,t} \neq 0$, then
$z_{i,t}$ changes in the next step with probability at least $1/(2d-1)$,
and conditioned on the $i$th dimension being the value that changes,
it decreases with probability at least $1/2 + 1/(8d-4)$.  
If $z_{i,t} = 0$, it increases in the next step with probability at most $2/(d+1)$.
\end{lemma}
\begin{proof}
This follows directly from the description above.  We note that worst
case with regard to the bias is when one dimension has $z_{i,t} \neq
0$ and $z_{j,t} \neq 0$ for $j \neq i$.  In this case, the only bias
that favors the $i$th dimension decreasing rather than increasing
stems for both choices being in the $i$th dimension.  With probability
$2(d-1)/d^2$ the pebble is chosen to move in the $i$th dimension and
some other dimension, in which case the pebble is equally likely to
move closer or further to the target.  With probability $1/d^2$
dimension $i$ is chosen for both moves, in which case it moves closer
with probability $3/4$.  Hence, conditioned on $z_{i,t}$ changing, it
increases with probability
$$\frac{(d-1)/d^2 + (3/4)/d^2}{2(d-1)/d^2 + 1/d^2} = \frac{d-1/4}{2d-1} = \frac{1}{2} + \frac{1}{8d-4}$$
as claimed.  

With regard to which dimension moves, we notice that when $z_{i,t}
\neq 0$, the probability of moving is least when we are on the
boundary in the $i$th dimension (and hence there is just one move in
that dimension), and other dimensions are not.  When $z_{i,t} = 0$,
that dimension is most likely to move if all others are on the boundary.
The above bounds reflect these cases.
\end{proof}
Our multi-dimensional biased random walk has a natural interpretation
as a discrete time queueing system, where customers arrive and wait at
a randomly chosen queue, where the arrival rate is slightly smaller
than the departure rate (except when a queue is empty).  In this
setting, our question concerns the time until the system empties from
a given starting state.  Surprisingly, despite this connection, we
could not find a statement corresponding to our desired result in 
the literature.  

The following follows easily from the bias shown above.
\begin{lemma}
\label{lem:step2}
In each individual dimension, if $z_{i,0}$ is bounded above by $n$, 
then with probability $1-O(1/n^{d+1})$, $z_{i,t}$ hits 0 in $O(d^2n)$
steps.  
\end{lemma}
\begin{proof}
Let $X$ be the number of steps taken in the $i$th dimension over the first
$64d^2n$ steps. Then by Lemma~\ref{lem:step1}, $\E[X] > 32dn$, and using a Chernoff bound
(e.g., \cite[Exercise 4.7]{MU})
\begin{eqnarray*}
\Pr(X \leq 16dn) \leq (e/2)^{-8dn}.
\end{eqnarray*}
The expected difference between the number of steps that decrease $z_{i,t}$
and the number of steps that increase $z_{i,t}$ grows with $X$, so we pessimistically
condition on $X \geq 16dn$.  Suppose that in this case that $0$ is never reached,
so the bias remains in effect over all $16dn$ steps.
Let $Y$ be number of decreases in the first $16dn$ steps.  
Then $\E[Y] \geq 8dn + 2n$, and again using a Chernoff bound (e.g., \cite[Exercise 4.7]{MU})
\begin{eqnarray*}
\Pr(Y \leq 8dn+n~|~X \geq 16dn) \leq \left (\frac{e^{-\delta}}{(1-\delta)^{(1-\delta)}} \right)^{-8dn+2n},
\end{eqnarray*}
where $\delta = 1/(8d+2)$.
Notice that if $Y \geq 8dn+n$ then in fact 0 was reached.  Hence the total probability that 0 is not reached is bounded above by 
$$\Pr(Y \leq 8dn+n~|~X \geq 16dn) + \Pr(X \leq 16dn) \leq (e/2)^{-8dn}$$
which is clearly $O(1/n^{d+1})$.  
\end{proof}

Similarly, the following result is standard for biased random walks.
\begin{lemma}
\label{lem:step3}
Once $z_{i,t}$ hits 0, with probability $1-O(1/n^{d+2})$, it will remain 
below $c_d \ln n$ for some constant $c_d$ (depending on $d$) over the
next $100d^2n^2$ steps.  
\end{lemma}
\begin{proof}
We may pessimistically assume that all $n^2$ steps are performed in the $i$th
dimension.  By the natural coupling we have that the probability that $z_{i,t}$ reaches
$c_d \ln n$ in exactly $k$ steps after it hits 0 is less than the probability in equilibrium
that a biased random walk with probability $\frac{1}{2} + \frac{1}{8d-4}$ is at $c_d \ln n$.  
(Technically, the biased random walk doesn't have an equilibrium distribution, because of
parity;  it will be an even number of steps from its starting point after an even number
of steps.  We can add an arbitrarily small self-loop probability and increase the number
of steps accordingly;  we use $n^2$ steps and assume the appropriate equilibrium distribution
for convenience.)  The equilibrium distribution for this biased random walk,
where  $\pi_j$ is the probability of being at $j$ in equilibrium, is 
easily found by detailed balance equations, which yield
$$\pi_j = \frac{2}{4d-1} \left( \frac{4d-3}{4d-1} \right)^k.$$
As $\pi_j$ is geometrically decreasing, for $j = c_d \ln n$ for some constant $c_d$
we have $\pi_k$ will be 
less than $1/(100d^2n^{d+4})$, and hence over the $100d^2n^2$ time steps we see $z_{i,t}$ never
reaches $c_d \ln n$ with probability $1-O(1/n^{d+2})$.
\end{proof}

We now use the following lemma, which in a slightly different
form appears in Theorem 7 of~\cite{}.  
We sketch the proof for completeness.
\begin{lemma}
\label{lem:finishgrid}
Starting from $(z_{1,t},z_{2,t},\ldots,z_{d,t})$, where each $z_{i,t}$ is 
at most $c_d \ln n$,   with probability $\Omega(1/(\ln \ln n)^c)$
for some constant $c$,
there is a $u = \alpha \ln n$ such that 
for some $k \leq u$, $z_{i,t+k}=0$ for all $i =1,\ldots,d$.
\end{lemma}
\begin{proof}
The analysis is broken into phases.  The first phase is of length
$O(\ln n)$, the second phase is of length $O((\ln n)^{1/2})$, and 
the $j$th phase is of length $O((\ln n)^{1/2^{j-1}})$.  All phases
have the same basic structure, except for the last.  Within each
phase, there are $d$ subphases.  In the $i$th subphase, we assume
that the $i$th coordinate $z_{i,t}$ moves according to a biased random walk
as previously described;  in all other phases, we pessimistically
assume that it moves according to an unbiased random walk.  Our goal
is to show that in each phase, each dimension moves closer to zero;
in particular, after the $j$th phase, all the $z_{i,t}$ are at
$O((\ln n)^{1/2^j})$ with constant probability.  This is shown using
Chernoff bounds (see Theorem 7 of~\cite{} for the corresponding calculation).  It follows that after 
$O(\ln \ln \ln)$ phases each $z_{i,t}$ will be bounded by a constant
with probability $\Omega(1)^{O(\ln \ln \ln n)} = \Omega(1/(\ln \ln n)^c)$.
At that point, there is a constant probability that the target vertex will
be reached in a constant number of steps.  
\end{proof}

Putting this all together now yields the theorem 
that the cover time for the 2-cobra random walk on $[0,n]^d$ is
$O(n)$ where the order notation hides constant factors that can depend on $d$.
Lemma~\ref{lem:step2} shows that for each dimension $i$, 
$z_{i,t}$ hits 0 within the first $64d^2n^2$ steps with high probability.  
Hence, by Lemma~\ref{lem:step3}, after $64d^2n^2$ steps all $z_{i,t}$ are
at most $c_d \ln n$ with high probability.  
We are therefore in a state where we reach the target vertex within an
additional $O(\ln n)$ steps with probability $\Omega(1/(\ln \ln n)^c)$.  If we have not reached the target vertex, however, we are still
within $c_d \ln n$ distance in each dimension with high probability.
That is, let ${\cal E}_1$ be the event that we did not reach the target vertex
within the $\alpha \ln n$ steps of Lemma~\ref{lem:finishgrid}, and let 
${\cal E}_2$ be the event that some $z_{i,t}$ is more than $c_d \ln n$
after those $\alpha \ln n$ steps.  We have by Lemma~\ref{lem:step3} that this 
is a low probability event (at least up through an additional $36d^2n^2$ steps),
so we can consider a polylogarithmic number of repeated trials of $\alpha \ln n$ steps.
As each trial succeeds with probability $\Omega(1/(\ln \ln n)^c)$,
we can conclude that we hit the target vertex with
probability at least $1-O(1/n^{d+1})$ within $O(n)$ steps.  It follows via a union bound that
with probability $1-O(1/n)$ all vertices are hit within $O(n)$ steps,
from which it readily follows that the cover time is $O(n)$.  
\end{proof}

We remark that this proof, while achieving $O(n)$ bounds, appears
quite loose in the constant factors.  As our proof does not directly take
advantage of the large number of pebbles within the system, we believe
the bounds could be tightened with respect to the dependence on $d$.
Even with this bound, we expect there remains further work
to fully understand the behavior of cobra walks on grids.  

We also remark that the multi-step case analysis used in
Lemma~\ref{lem:22d} can similarly be used to show that 2-cobra walks
on $k$-ary trees have cover times that are proportional to the graph's
diameter when $k = 2$ or $k = 3$.  We conjecture that in fact the
cover time for 2-cobra walks on $k$-ary trees is proportional to the
diameter for every constant $k$, where the constant of proportionality may
depend on $k$, similar to Theorem~\ref{thm:dgrid} for grids.

\section{Cover time for a graph with arbitrary conductance} 
In this section, we significantly improve and extend the results first developed in~\cite{Dutta:2015:CRW:2821462.2817830}, which provided 
an $O(\log^2 n)$ bound for the cover time of a $k$-cobra walk on a $d$-regular expander with 
an expansion only achieved by graphs such as random $d-$regular expanders and Ramanujan expanders. Here, we 
provide the first known bound for the cover time of cobra walks on a $d$-regular graph of \textit{arbitrary} conductance $\Phi$. While the upper bound 
is not useful for graphs with very low conductance, there are a wide class of graphs beyond expanders for which this guarantees 
rapid coverage,, e.g. the hypercube, power-law graphs, and random geometric graphs. For this section and the rest of the paper, 
we work exclusively with $2$-cobra walks, and use cobra walks to mean $2$-cobra walks where the meaning is clear.

\begin{theorem}
\label{ConductanceTheorem} 
Let $G$ be a bounded-degree, $d$-regular graph with conductance
$\Phi_G$. Then a cobra-walk starting at any vertex will cover $G$ in
$O(d^4 \Phi_G^{-2} \log^2(n))$ rounds, with high probability.
\end{theorem}

From this general bound, we have the following corollary, which corresponds to the previous result of 
~\cite{Dutta:2015:CRW:2821462.2817830}.

\begin{corollary}
\label{ExpanderCor} 
Let $G$ be a bounded-degree $d$-regular $\epsilon$-expander graph. Then a cobra walk starting at any vertex 
$v$ will cover $G$ in $O(\log^2(n))$ rounds with high probability when $d \in O(1)$. 
\end{corollary}

Due to the extreme difficulty of analyzing the progress of a cobra
walk explicitly, we follow~\cite{Dutta:2015:CRW:2821462.2817830} and
analyze a process that, while conceptually similar to a cobra walk,
has more structured rules which allow us to analyze walks taken by
individual pebbles with only limited dependence on one
another. Furthermore, this process stochastically dominates the cobra
walk when starting from the same vertex with respect to the time to
cover all of the vertices. Any upper bound on the cover time for this
process therefore automatically applies to the cover time of a cobra
walk as well.

This process, which we refer to as $W_{alt}$, can be defined as
follows: We start with $\delta n$ pebbles for some constant $\delta
\leq 1/2$, distributed arbitrarily among the vertex set $V$.
Furthermore, we assume that the pebbles have a total ordering, and
that each pebble knows its position in the ordering.  In this process,
unlike in the cobra walk, no pebbles split or coalesce -- the total
number of pebbles is an invariant.  Pebbles interact with one another
according to two simple rules. For each time step:
\begin{enumerate}
\item If one or two pebbles are co-located in time and space: at vertex $v$ at time $t$, 
each pebble chooses a random neighbor from $N(v)$ independently u.a.r and moves to that vertex.
\item If three or more pebbles are at $v$ at time $t$, the two pebbles with the lowest order each pick a vertex
independently from $N(v)$ u.a.r. and move to their chosen vertex. Label these vertices $u,w$ (keeping in mind 
that $u=w$ is allowed). The remaining pebble(s) at $v$ then each independently pick $u$ with probability $1/2$
or $w$ with probability $1/2$ and move to the vertex they have chosen.
\end{enumerate} 

\onlyLong{
The process $W_{alt}$ can be viewed,  at a single step, as a coalescing random walk in which the threshold for 
coalescence is three pebbles at the same vertex, rather than the standard two. As an added condition, the 
third and higher pebbles at a vertex (w.r.t. to the total ordering of pebbles) chooses which of the first two pebbles
to coalesce with via an unbiased coin flip.
}

If we are observing, for a single time step, a vertex $v$ at which two
or more pebbles (or zero, trivially) have landed, we would be unable
to distinguish between a cobra walk and a $W_{alt}$ process. On the
other hand, if we observe a vertex $v$ at which a single pebble has
landed, we \textit{would} be able to distinguish. In $W_{alt}$, in the
next step, the pebble at $v$ will act like a simple random walk and
move to a single neighbor. On the other hand, in the cobra walk, there
is some probability $p$ two neighbors will receive a pebble from $v$,
and probability $1-p$ only one neighbor will. Thus, the active set of
$W_{alt}$ can be viewed as a (possibly proper) subset of the active
set of a cobra walk when both are started from the same initial
state. Therefore, at any future time $t$, the size of the active set
of a cobra walk (viewed as a random variable) stochastically dominates
the size of the active set of $W_{alt}$. We can then "invert" this
argument to show that the cover time of $W_{alt}$ stochastically
dominates the cover time of the cobra walk. \onlyShort{ For a formal
  proof, we refer to the appendix, which contains the full version of
  the paper.}

Finally, for technical reasons, we make the $W_{alt}$ process a \textit{lazy} process. That is, at each step, with probability $1/2$ all 
pebbles remain in their same position. With probability $1/2$ a step proceeds with the probabilities described above. (Thus, to 
obtain the unconditioned probabilities of any particular action, we need to multiply the above probabilities through by $1/2$.)

\onlyLong{
\begin{lemma} 
\label{stochasticdominance}
Let $G$ be a $d$-regular graph. Let $S \subset V$ be a subset of the vertices of $G$ such that $|S| < n/2$. 
Consider $C$, a cobra walk which begins at all the vertices of $S$, and $W$, a $W_{alt}$ process which begins 
at all the vertices of $S$ and in which we are allowed to place an arbitrary number of pebbles at each $v \in S$, both 
at time $t=0$. Let $\tau_{C(S)}$ be the first time all the vertices are covered by $C$ 
and let $\tau_{W(S)}$ be the first time all the vertices are covered by $W$.
Then there exists a coupling under which 
$\tau_{C(S)} \leq \tau_{W(S)}$. 
\end{lemma}

\begin{proof}
Without loss of generality, let us assume that the initial
configuration of $W$ is such that no $v \in S$ has only one
pebble. Define a sequence $K_0,K_1,\ldots$ associated with $C$ and
$K'_1,K'_2,\ldots$ associated with $W$ where $K^{'}_i$ is the set of
vertices that have been covered by $C$ ($W$, respectively) at time
$i$. $G$ is covered when $K_T = n$. With each series we associate
another two series $\Delta_{(0,1)},\Delta_{(1,2)},\ldots$ and
$\Delta'_{(0,1)},\Delta'_{(1,2)},\ldots$, where $\Delta_{(i,j)}$
represents $|K_j| - |K_i|$. Note that unlike the sequence of active
sets of each process, the $K$ and $\Delta$ series are monotonically
non-decreasing.

The first time $\tau_C$ that all vertices for $C$ are covered is the time at which$\sum_{i=0}^{\tau_C} \Delta_{(i,i+1)} = n$,
and similarly for $\tau_W$. We 
now show that $C$ dominates $W$ statewise for each $\Delta_{(i,i+1)}$ , $\Delta'_{(i,i+1)}$. Note that, as a 
random variable, the distribution of $\Delta_{(0,1)}$ and $\Delta'_{(0,1)}$ are exactly the same, since we 
stipulated that every $v \in S$ for $W$ has more than one pebble. However, considering step $(1,2)$, we have 
that $\prob{\Delta'_{(1,2)} \geq c} \leq \prob{\Delta_{(1,2)} \geq c}$ by the simple fact that in the $(0,1)$ step 
there 
was a chance in $W$ that the group of pebbles of at least one of the vertices of $S$ would do the following: let 
$w$ and $x$ be the two nodes picked by the pebbles of $v$ to walk to during $(0,1)$. If $x=w$ then $C$ and $W$ are equivalent. 
If $x \neq w$, then with some finite probability (proportional to the number of pebbles at $v$), all 
but one of the pebbles would go to $w$, and only one pebble would go to $x$ (or vice versa). W.l.o.g. assume $x
$ receives only one pebble. Thus in the round $(1,2)$, whereas any $C$ in the exact same state has a non-zero 
probability of creating two pebbles from $x$, $W$ has zero chance of doing this and can no longer mimic $C$. 

Since $\prob{\Delta'_{(1,2)} \geq c} \leq \prob{\Delta_{(1,2)} \geq c}$, it follows that $\prob{\Delta'_{(i,j)} \geq c} 
\leq \prob{\Delta_{(i,j)} \geq c}$ for all $(i,j)$ by induction, since for each additional step $K'_i$ probabilistically 
occupies a smaller set of vertices than $K_i$ and we can again apply the reasoning above. (Note that just 
observing this occurs on the step $(1,2)$ is sufficient to prove the claim). Therefore, using the statewise 
stochastic dominance argument, it follows that $\prob{\tau_W(S) \geq K}  \geq \prob{\tau_C(S) \geq K}$.
\end{proof}
}

We are now ready to prove Theorem~\ref{ConductanceTheorem}. 
\onlyShort{This proof builds on the machinery developed
  in~\cite[Theorem 16 and Lemma 17]{Dutta:2015:CRW:2821462.2817830}.
  Here we will outline the proof technique, making note of the new
  ideas needed to provide a general bound on the cover time in terms
  of the conductance of the graph.  We have refactored the full proof
  for completeness, and it is available in the appendix.}
\onlyLong{ This proof uses the machinery developed in the proofs
  of~\cite[Theorem 16 and Lemma 17]{Dutta:2015:CRW:2821462.2817830}.
  We first make note of the new ideas needed to provide a general
  bound on the cover time in terms of the conductance of the graph.
  For completeness, we have refactored the proof with the changes
  necessary to prove our conductance claim and present it in its
  entirety here.}

One significant difference with previous work is the following.  The
previous analysis consisted of two stages, with the first stage
providing an exponential growth in the number of pebbles, but
requiring the graph to have extremely high expansion.  In contrast,
we begin our analysis of $W_{alt}$ with \textit{a large number} of
pebbles, all of which are located at a single initial, arbitrary
vertex $v$, and compare this to a cobra walk that starts at
$v$. \onlyLong{In~\cite{Dutta:2015:CRW:2821462.2817830}, the analysis
  was broken up into two stages. In the first stage, a cobra walk
  process was analyzed directly and it was shown that after $O(\log
  n)$ rounds, the size of the cobra walk went from $1$ vertex in the
  active set to $\delta n$ vertices in the active set, with high
  probability. However, one restrictive requirement was that the
  (vertex) expansion $\epsilon$ of the graph $G$ needed to be
  extremely high, satisfying the inequality:
\begin{equation*}
\dfrac{1}{\epsilon^2(1-\delta) + \delta} > \dfrac{d(de^{-k} + (k+1)) - \frac{k^2}{2}}{d(e^{-k} + (k-1)) - \frac{k^2}{2}}
\end{equation*}
where $k$ is the branching factor of the cobra walk ($2$, for our purposes) and $d$ is the degree of $G$. Only some families of 
expanders, such as Ramanujan graphs and random regular graphs with high degree, would
satisfy this condition. Once the cobra walk reaches $\delta n$ active vertices, we replace the cobra walk with a $W_{alt}$ in which 
we position one $W_{alt}$ pebble at each vertex that was active in the cobra walk at the time at which we perform the swap. The 
analysis then proceeds from this point to show that every vertex will have been visited by at least one pebble of $W_{alt}$ with 
high probability in $O(\log^2 n)$ time. 

}Since the stochastic dominance of the cover time of $W_{alt}$ over
a cobra walk holds for all starting distributions of the pebbles on
the vertices of $G$, it also therefore holds for the starting
distribution in which all $\delta n$ pebbles begin at the same
vertex. This allows us to work exclusively with $W_{alt}$ for the
entire analysis, bypassing the analysis of the cobra walk as it grows
from one active node to a linear number of active nodes and therefore
dropping the requirement that $G$ have extremely high expansion.

The second contribution of our new analysis is the derivation of
bounds on the probability of coverage of vertices in terms of the
conductance of the graph.  As a consequence, we are able to derive
bounds that hold for regular graphs with arbitrary values of $d$ and
arbitrary $\Phi_G$.

\begin{LabeledProof}{Theorem~\ref{ConductanceTheorem}}
To prove the theorem, we break up $W_{alt}$ into epochs of length $s$, where $s = f(\Phi_G,n)$. As we will see, 
the function $f$ has the form $f(\Phi_G,n) = O(\Phi_G^m \log n )$, for $m$ a constant. We show that each vertex 
$v$ has a constant probability of being hit by at least one pebble at the exact time step the epoch ends. We then 
go through $O(\log n)$ epochs to boost the probability $v$ is covered to a sufficiently high probability and then 
take a union bound over every vertex to obtain the result. 

We first prove that in a particular epoch vertex $v$ will be hit by at least one pebble with constant probability. Define $E_i$ to be 
the event that pebble $i$  for $i \in {1,\ldots,\delta n}$ covers an arbitrary vertex $v$ at time $s$. Then the event that $v$ is hit by at 
least one pebble is $\bigcup E_i$. Note that for any two pebbles $i,j$ it is not safe to assume that $E_i$ and $E_j$ are 
independent, 
since $i$ and $j$ may cross paths during their walks and have their transition probabilities affected by the rules of $W_{alt}$. 
However, we can use a second-order inclusion-exclusion inequality to lower-bound the probability:
\begin{equation*}
\prob{\bigcup_i E_i} \geq \sum_i \prob{E_i} - \sum_{i,j:i\neq j} \prob{ E_i \cup E_j}
\end{equation*}
Thus, we need to show that, for the r.h.s., the first quantity is a constant and the second quantity is smaller. 

\onlyShort{The first part is fairly straightforward and makes use of well-known results regarding simple random walks on graphs of a given conductance. We point the reader to the full version of the proof.} 

\onlyLong{

The first part is fairly straightforward. $\prob{E_i}$ can be viewed as the (marginal) probability that a simple random walk of pebble 
$i$ hits $v$ as time $s$. Indeed, if we only observe the movement of $i$ and ignore all other pebbles, its transition probability to 
any adjacent vertex from its current vertex is always $1/d$, regardless of the number of pebbles at the current vertex. Since this 
reduces to analyzing a simple random walk, we can use the well-known result that after $O(\log n / f(\Phi_G))$ time the probability 
that $i$ will be at $v$ will be within a $\pm 1/2n$ interval around $1/n$, assuming that $G$ is regular and the stationary distribution 
is the normalized uniform vector. For example, we can use the following result from~\cite{Spielman} to bound the maximum 
difference between components of the probability distribution of the walk after $t$ steps and the stationary distribution:
\begin{equation*}
|p_t(v) - \pi(v) | \leq \sqrt{\dfrac{d(v)}{d(u)}} e^{-t \nu_2} \leq e^{-t \Phi_G^2/2},
\end{equation*} 
where the quantity in the square root is $1$ because of the regularity of the graph, and we have substituted $\Phi^2_G/2$ for $
\nu_2$, the second largest eigenvalue of the normalized Laplacian of $G$. Thus, for $t > \dfrac{2\log(2n)}{\Phi_G^2}$, $ \prob{E_i} \leq 1/2n$.

}

Next we establish an upper bound for $\prob{E_i \cap E_j}$ using the joint (dependent) walks of pebbles $i$ and $j$.  Without loss
of generality, assume that $i$ has a lower order than $j$, and that if $i$ and $j$ are co-located in time and space, then not only is $i$ the lower 
order, but $j$ must behave like a third-or-greater priority pebble and choose the same vertex that $i$ next hops to with probability 
$1/2$. (Note that under this assumption, the total probability that $j$ moves to the same vertex as $i$ is $1/2 + 1/2d$, a fact which 
will be important shortly). 

We can view the random walks of $i$ and $j$ as a random walk over a graph with the topology of the tensor product graph $
\Gtensor$. The tensor product graph has the Cartesian product $V(G) \times V(G)$ as its vertex set, and an edge set defined as 
follows: vertex $(u,u') \in V(\Gtensor)$ has an edge to $(v,v') \in V(\Gtensor)$ if and only if $(u,v), (u',v') \in E(G)$. The random walk 
we construct on $\Gtensor$ is slightly different than a simple random walk on the same graph: we make the edges directed, and 
we attach weights to them such that the walk on the directed graph $D(\Gtensor)$ is isomorphic to the movement of pebbles $i$,$j
$ in a $W_{alt}$ on $G$. 

We will show in the next Lemma that the walk on $D(\Gtensor)$ has a stationary distribution, that this stationary distribution is 
close to $1/n^2$, and that the walk converges rapidly to it. Thus, after $s$ steps, the probability that $i$ and $j$ are at the same 
vertex in $G$ is bounded from above by $2/(n^2 + x) + 1/n^4$. 

Once we have established bounds for $\prob{E_i}, \prob{E_i \cap E_j}$, we can apply them to the full expression and get:
\begin{eqnarray*}
\prob{\bigcup E_i }
 &\geq& \sum_i \prob{E_i} - \frac{1}{2} \sum_{i \neq j} \prob{E_j \cap E_i } \\
&\geq& \delta n \frac{1}{2n} - \frac{1}{2}  \binom{\delta n}{2}\left( \frac{2}{n^2 + n}  + \frac{1}{n^4}\right) \\
&\geq& \frac{\delta}{2} - \frac{\delta}{4} (\delta n^2 - n)  \left( \frac{2}{n^2 + n}  + \frac{1}{n^4}\right) \\
&\geq& \frac{\delta}{2} - \frac{2 \delta^2}{4}
\end{eqnarray*}

Finally, to complete the proof, we have the following lemma:
\begin{lemma}
Let $G$ and $s$ be defined as in Theorem~\ref{ConductanceTheorem}. Let $i,j$ be two pebbles walking on $G$ according to the 
rules of $W_{alt}$ If $E_i \cap E_j$ are defined as the event in which $i$ and $j$ are both located at (arbitrarily chose) vertex $v$ 
at time $s$, then $\prob{E_i \cap E_j} \leq \dfrac{2}{n^2 + 2} + \dfrac{1}{n^4}$. 
\end{lemma}

\onlyShort{ Here we provide a brief sketch of the proof, full details of which are provided in the long version. We first create a directed graph $\D(\Gtensor)= \D$ through some minor modifications of the edge set of $\Gtensor$. These changes are such that the transition probability matrix of $\D$ is isomorphic to $W_{alt}$ on $G$.  We make use of the fact that $\D$ is an Eulerian digraph to calculate the stationary distribution $\pi$ of the walk on $\D$. We then show that $\Phi_G$ = $\Phi_{\Gtensor}$ and, making use of the work of~\cite{Chung05laplaciansand}, demonstrate through the use of the directed Cheeger inequalities that the second smallest eigenvalue of the normalized Laplacian of $\D$ is bounded from below by $O(\Phi_G^{-2})$. We then apply Theorem 7.3 of~\cite{Chung05laplaciansand} to show that this walk converges in $O(\Phi^{-2} \log n)$ time to the stationary distribution of $\Theta(1/n^2)$ for each vertex in $\D$. }

\onlyLong{

\begin{proof}
We first make a few observations about the topology of graph $\Gtensor$. There are two classes of vertices: the first class involves 
all vertices that have the form $(u,u)$ (and recall that $u \in  V(G)$). A pebble at $(u,u)$ would correspond to two pebbles 
occupying $u$ in the paired walk on $i,j$ on $G$. Label the set of such vertices $S_1$. The cardinality of $S_i$ is $n$. The 
remaining vertices, which belong to what we label $S_2$, are of the form $(u,v)$, for $u\neq v$. There are $n^2 - n$ such vertices. 
Also note that in the \textit{undirected} graph $\Gtensor$, each vertex has degree $d^2$. Further, every vertex in $S_1$ will have 
$d$ neighbors also in $S_1$, by virtue of the regularity of $G$. 

Many of the spectral properties of $G$ apply to $\Gtensor$. Because $G$ is an $\epsilon$-expander, the transition probability 
matrix of a simple random walk on $G$ has a second eigenvalue $\alpha_2(G)$ bounded away from zero by a constant. It is well 
known (c.f.~\cite{LevinPeresWilmer2006}) that the transition probability matrix of a simple random walk on $\Gtensor$ will have $
\alpha_2(\Gtensor) = \alpha_2(G)$, which we will refer to as $\alpha_2$ henceforth.

Next, we transform the (undirected) graph $\Gtensor$ into a directed graph $D(\Gtensor)$ according to the following rules: For 
every undirected edge $(x,y)$ in $\Gtensor$, we replace it with $2$ directed edges: $x \rightarrow y $ and $y \rightarrow x$. As 
mentioned, every vertex in $S_1$ will have $d$ neighbors also in $S_1$, meaning that there will be one directed arc $x \rightarrow 
y$ for every vertex $x \in S_1$ and $y \in N(x)$ for $y$ in either $S_1,S_2$. We next add an additional $d$ copies of edge $x 
\rightarrow y$ for every such original edge.  It is relevant for the analysis to note that the regularity of the subgraph of $\Gtensor$ 
induced on $S_1$ implies that every vertex in $\D(\Gtensor)$ will have an equal number of out and in arcs. Therefore $
\D(\Gtensor)$ is an Eulerian digraph, a fact which will be important when calculating the stationary distribution of a walk on this 
graph.

Next, we calculate the transition probabilities of a random walk on $\D(\Gtensor)$. An outgoing edge from $x$ is picked with 
probability $1/d$. However, when we compute the probability of transitioning to a neighboring vertex, there is a difference 
between nodes in $S_1$ and $S_2$. In $S_2$, each transition occurs with probability $1/d^2$, as in the case of a symmetric walk. 
Note that this includes transitions to nodes in $S_1$ as well as within $S_2$. On the other hand, the probabilities of transitions 
from nodes in $S_1$ are modified. The probability of transitioning from a vertex $x \in S_1$ to one of its $d^2 -d$ neighbors in 
$S_2$ is $1/2d^2$. The probability of transition to another neighbor in $S_1$ is $(d+1)/2d^2$ (on account of the multiple edges). 
Thus, the walk on digraph $\D(\Gtensor)$ corresponds is isomorphic to the joint walk of pebbles $i,j$ on $G$ according to the rules 
of $W_{alt}$.

Because the walk on $D(\Gtensor)$ is irreducible, it has a stationary distribution $\pi$, which is the (normalized) eigenvector of the 
dominant eigenvalue of the transition matrix $M$ of the walk on $D(\Gtensor)$. Furthermore, because $D(\Gtensor)$ is Eulerian, 
the stationary distribution of vertex $x$ is exactly given by: $d^+(x)/|E|$,  where $d^+(x)$
is the out-degree of $x$. Therefore there are 
only two distinct values of the components of the stationary vector: for all $x \in S_1$, $\pi(x) = 2/(n^2 + n)$, while for all $y \in 
S_2$, $\pi(y)= 1/(n^2 + n)$.

A note of caution: because $G$ and $\Gtensor$ have such nice spectral properties, and because $\D(\Gtensor)$ represents such 
a minor modification of $\Gtensor$ it would be natural to infer that $\D(\Gtensor)$ also has many of the same nice properties 
(particularly, the property of rapid convergence). However, it is often the case that properties of Markov chains on undirected 
graphs do not also hold for Markov chains on directed graphs. Thus, we must carefully check our rapid convergence hypothesis. 
Fortunately, following closely the work of~\cite{Chung05laplaciansand}, we can indeed verify that the walk on $\D(\Gtensor)$ 
converges rapidly to its stationary distribution.

For succinctness of notation denote $\D = D(\Gtensor)$. Consider the function $F_{\pi}: E(\D) \rightarrow \Re$ given by $F_{\pi}
(x,y) = \pi(x) P(x,y)$, where $\pi(x)$ is the $x$-th component of the stationary distribution of the walk on $\D$ and $P(x,y)$ is the 
associated transition probability moving from $x$ to $y$. Then $F_{\pi}$ is the \textbf{circulation} associated with the stationary 
vector as shown in Lemma 3.1 of ~\cite{Chung05laplaciansand}. Note that a circulation is any such function that satisfies a 
balance equation: $\sum_{u, u\rightarrow v} F(u,v) = \sum_{w,v\rightarrow w} F(v,w)$. 

There is a Cheeger constant for every directed graph, defined as:
\begin{equation}
h(G) = \inf_S \frac{F_{\pi}(\partial S)}{\min\{F_{\pi}(S),F_{\pi}(\bar{S})\}}
\end{equation}
where $F_{\pi}(\partial S) = \sum_{u\in S,v\notin S} F(u,v)$, $F(v) = \sum_{u,u\rightarrow v} F(u,v)$, and $F(S)=
\sum_{v \in S} F(v)$ for a set S. Furthermore, Theorem 5.1 of~\cite{Chung05laplaciansand} shows that the second eigenvalue, $\lambda$, of the 
Laplacian of $\D$ satisfies:
\begin{equation}
2 h(\D) \geq \lambda \geq \frac{h^2(\D)}{2}
\end{equation}

The Laplacian of a directed graph is defined somewhat differently than the Laplacian of an undirected graph. However, because 
we will not use the Laplacian directly in our analysis, we refer the reader to~\cite{Chung05laplaciansand} for the definition. We will 
directly bound the Cheeger constant for $\D$, and hence produce a bound on the second eigenvalue of the Laplacian. This second 
bound will then be used to provide a bound on the convergence of the chain to its stationary distribution. 

First, w.l.o.g., assume that $F_{\pi}(S)$ is smaller than its complement.  Furthermore assume that $S$ is the set that satisfies the $
\inf$ condition in the Cheeger constant. We have $F_{\pi}(\partial S)  = \sum_{x \rightarrow y, x \in S, y \in \bar{S}} \pi(x) P(x,y)$, 
and \\  $F_{\pi}=\sum_{x \rightarrow y, x \in S} \pi(x) P(x,y)$. The first sum occurs over all (directed) edges that cross the cut of $S$, 
while the second sum occurs over all edges leaving vertices in $S$: the numerator and denominator of the conductance $\Phi_{G
\times G}$. Thus we can provide a lower bound for the entire Cheeger constant: 
\begin{eqnarray*}
h(\D) &=& \inf_S \frac{F_{\pi}(\partial S)}{\min\{F_{\pi}(S),F_{\pi}(\bar{S})\}} \\
&\geq& \Phi_{G \times G} \dfrac{P_{min} \pi_{min}}{P_{max} \pi_{max}} \\
&=& \Phi_{G \times G} \cdot \dfrac{\dfrac{1}{4d^2} \dfrac{1}{(n^2 + n)}}{\dfrac{1}{2} \dfrac{2}{(n^2 + n)}} =  \dfrac{\Phi_{G \times G}}
{4d^2}  
\end{eqnarray*}  
where $P_{max}$ is $1/2$ because of the lazy property of $W_{alt}$.

Next, we show that $G$ and $\Gtensor$ have the same conductance. As we noted earlier, $G$ and $\Gtensor$
have the same second-smallest Laplacian eigenvalue $\nu_1$. From~\cite{TCS-010}, we know that they 
have the same spectral gap $\gamma$, 
defined as $1 - \lambda(G)$ where $\lambda$ is the second largest eigenvalue of the adjacency matrix of $G$. 
From Theorem 4.14, we also know that the spectral expansion (gap) of $G$ is equal to $\alpha \epsilon^2/2$, where 
$\alpha$ is the fraction of edges that are self loops and $\epsilon$ is the edge expansion of $G$, which for $d$-regular 
graphs is the same as the conductance $\Phi_G$. Thus, $\Phi_G = \Phi_{\Gtensor}$.

We then apply the lower (directed) Cheeger inequality to have $\lambda \geq \dfrac{\Phi_G^2}{32d^4}$. 

We now show rapid convergence (in logarithmic time in $n$) of the walk on $\D$ to the stationary distribution. To measure distance 
from the stationary distribution, we use the $\Xi$-square distance:
\begin{equation}
\Delta'(t) = \max_{y \in V(\D)} \left(  \sum_{x \in V(\D)} \dfrac{(P^t(y,x) - \pi(x))^2}{\pi(x)}\right)^2
\end{equation} 
It is straightforward to show that any distance in the $\Delta'$ metric is no smaller than a distance using a total-variational distance 
metric. Hence the distribution of a random walk starting anywhere in $\D$ will be close to its stationary distribution w.r.t. each 
component vertex.

We next apply Theorem 7.3 from~\cite{Chung05laplaciansand}.
\begin{theorem} 
Suppose a strongly connected directed graph $G$ on $n$ vertices has Laplacian eigenvalues $0=\lambda_0 \leq \lambda_1 \leq 
\ldots \leq \lambda_{n-1}$. Then G has a lazy random walk with the rate of convergence of order $2 \lambda_1^{-1} (-\log \min_x 
\pi(x))$. Namely, after at most  $t \geq 2 \lambda_1^{-1}((-\log \min_x \pi(x) + 2c)$  steps, we have: 
\begin{equation*}
\Delta'(t) \leq e^{-c}
\end{equation*}
\end{theorem}

Recall that we provided a lower bound for $\lambda$, the second-smallest eigenvalue of the Laplacian of $\D$ in terms of the 
conductance of $\Phi_G$. We can thus apply the above theorem to show that after at most:
\begin{equation*}
s = \dfrac{32d^4}{\Phi_G^2} \left(\log(n^2 + n) + 4 \log(n^2)\right)
\end{equation*}
steps we will have $\Delta'(t) \leq \dfrac{1}{n^4}$. For the random walk on $\D$, after a logarithmic number of steps $s$ (in $n$), a 
walk that starts from any initial distribution on $V(\D)$ will be within $n^{-4}$ distance of the stationary distribution of any vertex. 
\end{proof}

Mapping our analysis back directly to the coupled walk of $i,j$ on $G$, it then holds that $\prob{E_i \cap E_j} \leq 2/(n^2 + n) + 1/
n^4$ when pebbles $i$ and $j$ start from any initial position. 

}

\onlyShort{This concludes the proof of Theorem~\ref{ConductanceTheorem}. It follows immediately that when $d$ is a constant 
and $G$ is any expander with conductance $\Phi_G$ bounded away from zero by a constant, a cobra walk will cover $G$ in at 
most $O(\log^2 n)$ time.}

\end{LabeledProof}

\input{macros}
}
\newcommand{\degree}[1]{d(#1)}
\newcommand{\ws}{\widehat{\sigma}}

\section{Cover time for general graphs}
In this section, we show that the cover time of $2$-cobra walks for
general graphs is $O(n^{11/4} \log n)$.  For $\delta$-regular graphs,
we obtain an improved bound of $O(n^{2-1/\delta} \log n)$.  We
establish the bounds on the cover time by proving that for any start
vertex $u$ and any target vertex $v$, the hitting time from $u$ to
$v$, denoted by $H(u,v)$, is $O(n^{11/4})$ for general graphs and
$O(n^{2-1/\delta})$ for $\delta$-regular graphs.  We then apply
Matthews' Theorem (Theorem~\ref{matthewscobra}), which actually give a
high probability result.

To analyze the cobra walk, we consider biased versions of the standard
random walk; in each step, with some probability, instead of moving to
a neighbor chosen uniformly at random, the walk moves to a vertex as
determined by a memoryless controller whose aim is to reach a certain
target vertex.  The biased walk we analyze is a variant of {\em
$\eps$-biased walks}\/ defined in~\cite{azar}.  In
Section~\ref{sec:general.biased}, we review the $\eps$-biased walks
of~\cite{azar} and introduce and analyze our variant, which we call
{\em inverse-degree-biased} walks.  In
Section~\ref{sec:general.regular}, we use $1/\delta$-biased walks to
derive an $O(n^{2 - 1/\delta})$ bound on the hitting time of the cobra
walk on $\delta$-regular graphs, improving on the quadratic bound for
standard random walks.  In Section~\ref{sec:general.nonregular}, we
use inverse-degree-biased walks to derive an $O(n^{11/4})$ bound on
the hitting time of the cobra walk on nonregular graphs, improving on
the cubic bound for standard random walks.

\subsection{Biased random walks}
\label{sec:general.biased}
An {\em $\eps$-biased walk}\/ on $G$ is defined as follows: in each step,
with probability $1 - \eps$ one of the neighbors is selected uniformly
at random, and the walk moves there; with probability $\eps$, a
controller gets to select which neighbor to move to.  The controller
can be probabilistic, but it is time independent.  The primary
motivation for biased random walks is to study how much can a
controller increase the occupancy probabilities at a targeted subset
$S$ of the vertices in the stationary distribution.  A central result
of~\cite{azar} is the following lower bound achieved by an optimal
controller in an $\eps$-biased walk.
\begin{theorem}[Theorem~2 of~\cite{azar}]
\label{thm:azar}
Let $G = (V,E)$ be a connected graph, $S\subset V$, $v \in S$ and
$x \in V$. Let $\Delta(x,v)$ be the length of the shortest path
between vertices $x$ and $v$ in $G$ and $\Delta(x,S) = \min_{v\in
S} \Delta(x,v)$. Let $\beta = 1 -\epsilon$. There is an $\eps$-bias
strategy for which the stationary probability at $S$ (i.e. the sum of
the stationary probabilities of $v\in S$) is at least
\begin{equation}
\dfrac{\sum_{v\in S} \degree{v}}{\sum_{v \in S} \degree{v} + \sum_{x\notin S} \beta^{\Delta(x,S) -1} \degree{x}}.
\end{equation} 
\end{theorem} 
\junk{
For the special case of $$-regular graphs, the following corollary is
immediate.
\begin{corollary}{Corollary~1 of~\cite{azar}}
\label{cor:azar}
Let $G = (V,E)$ be any connected $d$-regular graph and let $S$ be a
subset of $V$.  Then there is an $\eps$-bias strategy for which the
sum of the stationary probabilities for vertices in $S$ is at least
$\left(\frac{|S|}{n}\right)^{1 - c\eps}$, for a constant $c > 0$
depending only on $d$.
\end{corollary}
}

The bias in an $\eps$-biased walk is constant ($\eps$) at each vertex.
To analyze the cobra walk, we define a biased walk in which the bias
available to the controller at each vertex of the graph varies with
the vertex.  Specifically, we define the {\em inverse-degree-biased
walk}\/ with target $x$: if the walk is at $x$, then the next vertex
visited is a neighbor of $x$ chosen uniformly at random; if the walk
is at vertex $v \neq x$, then with probability $1 - 1/\degree{v}$ one
of the neighbors is selected uniformly at random and the walk moves
there, and with probability $1/\degree{v}$ a controller gets to select
which neighbor to move to.  The goal of the controller is to minimize
the hitting time to $x$.

Our main motivation for introducing inverse-degree-biased walks is the
following dominance argument.  For any pair of vertices $u$ and $v$,
let $H^*(u,v)$ denote the smallest hitting time to $v$ achievable by
an inverse-degree-biased walk with target $v$, starting at $u$.
Recall that $H(u,v)$ denotes the hitting time at $v$ for a cobra walk
starting at $u$.

\begin{lemma}
\label{lem:dominate}
For any pair of vertices $u$ and $v$, $H(u,v) \le H^*(u,v)$.
\end{lemma}
\begin{proof}
Our proof is by a coupling argument.  Consider
the cobra walk starting at $u$ at time $0$.  For any vertex $x$ and
any time $t$, let $E^*_t(x)$ denote the event that $x$ is active in
the cobra walk at the start of round $t$.  For any vertex $x$, any
neighbor $y$ of $x$, and any time $t$, let $P^*_t(x,y)$ be the
probability that $y$ is active at the start of round $t+1$ conditioned
on the event that $x$ is active at the start of round $t$.  That is,
\[
P^*_t(x,y) = \Pr[E^*_{t+1}(y) \mid E^*_t(x)].
\]
We next consider any inverse-degree-biased walk starting from $u$.
Let $E_t(x)$ denote the event that the walk is at vertex $x$ at the
start of round $t$.  Let $P'_t(x,y)$ be the probability that the
inverse-degree-biased walk is at $y$ at the start of round $t+1$,
conditioned on the event that the walk is at $x$ at the start of round
$t$.  By the definition of the cobra walk and the
inverse-degree-biased walk, we have the straightforward derivation.
\begin{eqnarray*}
P^*_{t}(x,y) & = & 1 - \left(1 - \frac{1}{\degree{u}}\right)^2\\
& \ge & \frac{2}{\degree{u}} - \frac{1}{\degree{u}^2}\\
& = & \frac{1}{\degree{u}} + \left(1 - \frac{1}{\degree{u}}\right) \frac{1}{\degree{u}}\\
& \ge & P'_t(x,y).
\end{eqnarray*}

\junk{
that $y$ is
occupied in the cobra walk at the end of round $t$ is at least the
probability $P'_t(y)$ that the $x$-targeted walk is at $y$ at the end
of round $t$.  The proof is by induction on $t$.

The base case is $t = 0$.  The desired claim is trivially true since
both the walks start at $x$.  For the induction hypothesis, suppose
the claim holds for time $t$.  We now consider round $t + 1$ of the
two walks.  For an arbitrary vertex $v$ } 

We now invoke a standard coupling argument to obtain that for any
round $t$ and any vertex $x$, the probability that $x$ is active in
the cobra walk at the start of round $t$ is at least the probability
that the inverse-degree-biased walk is at $x$ at the start of round
$t$.  This also implies that the expected time for $u$ to become
active in the cobra walk is at most the expected time to hit $x$ in
the inverse-degree-biased walk.  Setting $x = v$ 
and choosing the inverse-degree-biased walk to be one that
minimizes the hitting time to $v$ yields $H(u,v) \le H^*(u,v)$.
\end{proof}

\junk{
We now split the analysis into two tracks.  For $d$-regular graphs, an
inverse-degree-biased walk is a $1/d$-biased walk, and hence we can
directly invoke Corollary~\ref{cor:azar}.  

For any vertex $x$, we define  walk as follows.  In
any step of the $x$-targeted walk, if the walk is currently at $u \neq
x$, then with probability $1/\degree{u}$, the walk proceeds to an
arbitrary vertex chosen by an arbitrary time-independent controller;
with probability $1 - 1/\degree{u}$, the walk proceeds to a vertex
chosen uniformly at random from the neighbors of $u$.  If the walk is
at $x$, then the walk proceeds to a vertex chosen uniformly at random
from the neighbors of $x$.
}

\subsection{$O(n^{2-1/\delta})$ hitting time for $\delta$-regular graphs}
\label{sec:general.regular}
For the case of $\delta$-regular graphs, the inverse-degree-biased
walk with target $v$ is exactly a $1/\delta$-biased walk with target
$v$.  We invoke Theorem~\ref{thm:azar} for $\delta$-regular graphs
with $\beta = 1 - 1/\delta$, and letting $S$ equal $\{v\}$.  We
calculate a lower bound on stationary probability at $v$ of the
$1/\delta$-biased walk with target $v$, we need to determine the
maximum value that $\sum_{x \neq v} \beta^{\Delta(x,v)-1}$ can take.
We calculate this maximum value, over all graphs with maximum degree
$\delta$, hence yielding an upper bound for $\delta$-regular graphs as
well.  Let $n_i$ denote the number of vertices within $i$ hops of $v$.
Then, we have
\begin{eqnarray*}
\sum_{x \neq v} \beta^{\Delta(v,x) -1} &=& \sum_{i \ge 1} n_i \beta^i.
\end{eqnarray*}
Since the degree is at most $\delta$, we have the following
constraints on the $n_i$s, given by the number of vertices within a
certain number of hops from $v$.
\[
\sum_{1 \le i \le j} n_i \le \min\left\{n-1, \delta\left(1 + \sum_{1 \le i < j} (\delta-1)^i\right)\right\}.
\]
Since $\beta < 1$, and the total number of vertices is bounded by $n$,
the sum $\sum_{i \ge 1} n_i \beta^i$ is maximized when
\begin{eqnarray*}
n_1 & = & \delta\\
n_j & = & \delta(\delta-1)^j \; 1 < j < L,
\end{eqnarray*}
where $L$ is the smallest $j$ such that $\delta\left(1 + \sum_{1 \le i < j}
(\delta-1)^i\right)$ is at least $n-1$.  Then, we have $n_L = n - 1
- \sum_{j < L} n_j$.  With elementary algebra, we calculate the value
of $L$ as
\begin{equation*}
L = \log_{(\delta-1)} \left[ \dfrac{(n-1)}{\delta} (\delta-2) + 1 \right].
\end{equation*} 

We now bound $\sum_{x\neq v} \beta^{\Delta(x,v) -1}$ as follows.
\begin{eqnarray*}
\sum_{x \neq v} \beta^{\Delta(v,x) -1} &=& \delta\beta^0 + (\delta-1)\delta \beta^2 + \delta(\delta-1)^2 \beta^2\\
& & + \ldots + \delta(\delta-1)^{L-1}\beta^{L-1} \\
&=& d \left( \sum_{i =0}^{L-1} ((\delta-1)\beta)^i\right) \\
&=& d \left[ \dfrac{((\delta-1)\beta)^L - 1}{(\delta-1)\beta -1 } \right] \\
&\leq& C \left[ ((\delta-1)\beta)^L \right] 
\end{eqnarray*}
where $C = \delta/ ((\delta-1)\beta - 1)$.   

With some minimal algebra, we can see that:
\begin{equation*}
(\delta-1)^{\log_{\delta-1} (n-1)(\delta-2)/\delta + 1} = (n-1)(\delta-2)/\delta + 1 <  n
\end{equation*} 
and that
\begin{eqnarray*} 
\beta^L &=&  (1-1/\delta)^{\log_{\delta-1}(n-1)(\delta-2)/(\delta-1) + 1}\\
& \leq & (\delta-1)^{(-1/\delta) \log_{\delta-1} ((n-1)(\delta-1)/(\delta-2) + 1)} \\
&=& \left[\dfrac{(n-1)(\delta-2) + \delta}{\delta}\right]^{-1/\delta} < n^{-1/\delta}.
\end{eqnarray*} 
Therefore, the stationary distribution at $v$ is at least $1/(1 +
n^{1-1/\delta})$, giving a return time to $v$ of at most $1 +
n^{1-1/\delta}$.  Using the tree-traversal argument and noting that it
will take, in expectation, $d$ returns to current vertex $v$ to make
one ``progress'' step in the traversal in expectation for a total cost
of $\delta n^{1 -1/\delta}$ per progress step. Thus, the total cost of
coverage will be $2\delta Cn^{2-1/\delta} = O(n^{2-1/\delta})$ for
constant $\delta$.
\begin{theorem}
\label{thm:regular}
Let $G$ be a finite, $\delta$-regular graph. Then the hitting time of
a cobra walk on $G$ is in $O(n^{2-\frac{1}{\delta}})$, for any
constant $\delta$.
\end{theorem} 

\junk{

{\bf MM:  This should be spelled out more clearly.  I don't get all the 
steps here.  I'm assuming what is meant is that we take a shortest path
from $u$ to $v$, and we make progress on average once every $d$ returns.
I'm not clear where the $C$ came back from in the expression below, or the 2 
in the $2dc$.  I'm also not clear why this works, since what we seem to really
want is the ``return time conditioned on having gone the wrong direction away
from the target'', and not the return time to the vertex.}

{\bf MM:  do we mean cover time or hitting time in the theorem below?}
}

\junk{
{\bf MM:  I'm pretty sure we should have $d(x)$ and not $d_x$ in the above
paragraph to be consistent with what comes previously -- please fix if that's
the case.}  

{\bf MM:  The paragraph below never made sense to me, and probably is what
needs to be rewritten.  I'm putting a rewritten form below but probably
best to scratch and rewrite with a cleaner argument?}
To see this, note that clearly $(1-\beta) < 1$ and hence is decreasing
for $(1-\beta)^y$ for $y > 1$. There grouping together all vertices at
distance $1,2,\ldots$, it is clear that the maximum is satisfied in
the case where each vertex has a maximum branching factor of $d-1$
away from $v$ and in which there are no cycles. Hence, a $d$-ary
balanced tree.

{\bf That is, the maximum for the sum is achieved when each vertex has $d-1$
children in the direction away from $v$, and where we pessimistically
assume that there are no cycles. This corresponds to a $d$-ary
balanced tree.}

{\bf MM:  I feel like we could just give the depth $L$ below without
bothering with the elementary calculation.}

We next calculate the depth of the tree, $L$. We have:
\begin{equation*}
n = 1 + d + (d-1)d  + (d-1)^2 d + \ldots + (d-1)^{L-1}
\end{equation*}
and hence:
\begin{equation*}
d \left[ \dfrac{(d-1)^L - 1}{d-2} \right] = n -1 
\end{equation*}
yielding
\begin{equation*}
L = \log_{(d-1)} \left[ \dfrac{(n-1)}{d} (d-2) + 1 \right].
\end{equation*} 
}

\subsection{$O(n^{11/4})$ hitting time for general graphs}
\label{sec:general.nonregular}
To bound $H(u,v)$ in a general graphs, 
let $P$ be a shortest-hop path
from $u$ to $v$, given by the sequence of vertices $u = u_0, u_1, u_2,
\ldots, u_{|P|} = v$.  Then
\begin{eqnarray}
\label{eqn:hit_u_v}
H(u,v) \le \sum_{i = 0}^{|P| - 1} H(u_i, u_{i+1}).
\end{eqnarray}
Let $R(x)$ denote the minimum return time to $x$ among all
inverse-degree-biased walks with target $x$.  We have
\begin{eqnarray}
\label{eqn:return}
R(x) & = & 1 + \sum_{y\in N(x)} \frac{H^*(y,x)}{\degree{x}}.
\end{eqnarray}
Furthermore, we have
\begin{eqnarray*}
H(u_i, u_{i+1}) 
& \le & 1 + \sum_{y \in N(u_i), y \neq u_{i+1}} H(y,u_{i+1})/\degree{u_i}\\
& \le & 1 + \sum_{y \in N(u_i), y \neq u_{i+1}} \frac{H(y,u_i) + H(u_i,u_{i+1})}{\degree{u_i}}\\
& \le & 1 + \sum_{y \in N(u_i), y \neq u_{i+1}} \frac{H^*(y,u_i) + H(u_i,u_{i+1})}{\degree{u_i}}\\
& \le & R(u_i) + \sum_{y \in N(u_i), y \neq u_{i+1}} H(u_i,u_{i+1})/\degree{u_i}\\
& = & R(u_i) + \left(1 - \frac{1}{\degree{u_i}}\right) H(u_i,u_{i+1}).\\
\end{eqnarray*}
(The second step follows from the fact that the hitting time from $y$
to $u_{i+1}$ is at most the sum of the hitting times from $y$ to $u_i$
and from $u_i$ to $u_{i+1}$.  The third step follows from
Lemma~\ref{lem:dominate}.  The fourth step follows from
Equation~\ref{eqn:return}.  The final equality holds since $u_i$ has
$\degree{u_i}$ neighbors.)

The above equation easily yields $H(u_i,u_{i+1}) \leq \degree{u_i} R(u_i)$,
which combined with Equation~\ref{eqn:hit_u_v} yields
\begin{eqnarray}
H(u,v) \le \sum_{i = 0}^{|P| - 1} H(u_i, u_{i+1}) \le \sum_{i = 0}^{|P| - 1} \degree{u_i} R(u_i). \label{eqn:hit}
\end{eqnarray}

We prove the following extension of Theorem~\ref{thm:azar}, applied to
inverse-degree-biased walks.  For the sake of completeness, we include
the proof, which essentially mimics the proof of
Theorem~2 in~\cite{azar}.
\begin{lemma}
\label{lem:inverse}
Let $G = (V,E)$ be a connected graph and $S$ be an arbitrary subset of
$V$.  For a given path $P$, let $\sigma(P)$ denote the product
$\prod_{y \in P} (1 - 1/\degree{y})$.  For $v \in S$ and $x \in V -
S$, let $\ws(x,v)$ denote the maximum, over all paths $P$ from $x$ to
$v$, of $\sigma(P)$.  Let $\ws(x,S)$ denote the minimum, over all
$v \in S$, of $\ws(x,v)$.  There is an inverse-degree-biased walk for
which the stationary probability at $S$ (i.e. the sum of the
stationary probabilities of $v\in S$) is at least
\begin{eqnarray*}
\frac{\sum_{v\in S} \degree{v}}{\sum_{v \in S} \degree{v} + \sum_{x\notin S} (\ws(x,S) \degree{x})}.
\end{eqnarray*} 
\end{lemma} 
\begin{proof}
Define a probability distribution $\pi^M$: $\pi^M(v)
= \gamma \degree{v}$ for $v \in S$ and $\pi^M(x)
= \gamma \widehat{\sigma}(x,S) \degree{x}$ for $x \notin S$, where
\[
\gamma = \frac{1}{\sum_{v \in S} \degree{v} + \sum_{x \notin S} (\ws(x,S) \degree{x})}.
\] 
Let $M$ be the transition probability matrix for which the stationary
distribution given by Metropolis
Theorem~\cite[Theorem~1]{azar}~\cite{metropolis} is $\pi^M$.  Define
the following transition probability matrix: $P_{x,x} = 0$ for all
$x \in V$, and for $y \neq x$, $P_{x,y} = M_{x,y}/(1 - M_{x,x})$.  We
now argue that $M_{x,y} \ge (1 - 1/\degree{x})/\degree{x}$ for any
neighbor $y$ of $x$.  By the Metropolis Theorem, $M_{x,y}$ is either
$1/\degree{x}$ or $\pi^M(y)/(\degree{y}\pi^M(x))$.  In the former
case, the desired lower bound on $M_{x,y}$ is trivial.  In the latter
case, we have
\[M_{x,y} = \frac{\ws(y,S)}{\degree{x} \ws(x,S)} \ge \frac{1 - 1/\degree{x}}{\degree{x}},
\]
where the last inequality follows from the fact that $\ws(y,S) \ge (1
- 1/\degree{x}) \ws(x,S)$.  We thus obtain $P_{x,y} \ge M_{x,y} \ge (1
- 1/\degree{x})/\degree{x}$ for any neighbor $y$ of $x$.  So $P$ is a
transition probability matrix of an inverse-degree-biased walk.

We next show that for each vertex in $S$, its stationary probability
under $P$ is at least its stationary probability under $M$.  Let
$H^M(x,v)$ and $H^P(x,v)$ denote the hitting time to $v$ under the
walks defined by $P$ and $M$, respectively, with start vertex $x$.
For each vertex $x \neq v$, we have
\[
H^M(x, v) = 1 + M_{x,x} H^M(x, v) + \sum_{y \in N(x)} M_{x,y} H^M(y,v),
\]
which yields
\begin{eqnarray*}
H^M(x,v) & = & \frac{1}{1 - M_{x,x}} + \sum_{y \in N(x)} \frac{M_{x,y}}{1 - M_{x,x}} H^M(y,v)\\
& \ge & 1 + \sum_{y \in N(x)} P_{x,y} H^M(y,v).
\end{eqnarray*}
By~\cite[Lemma~1]{azar}, for every $x \in V$ and $v \in S$,
$H^M(x,v) \ge H^P(x,v)$.  If we denote the return time to vertex $v$
in the walks under $P$ and $M$ by $R^P(v)$ and $R^M(v)$, respectively,
we obtain that
\begin{eqnarray*}
R^M(v) & = & 1 + \sum_{y \in N(v)} M_{v,y} H^M(y,v)\\
& = & 1 + \sum_{y \in N(v)} P_{v,y} H^M(y,v)\\
& \ge & 1 + \sum_{y \in N(v)} P_{v,y} H^P(y,v)\\
& = & R^P(v).
\end{eqnarray*}
We thus have
\[
\pi^P(S) \ge \pi^M(S) = \frac{\sum_{v \in S} \degree{v}}{\sum_{v \in S}\degree{v} + \sum_{x \notin S} \degree{x} \ws(x, S)}.
\]
\end{proof}
\begin{corollary}
\label{cor:inverse}
Let $v$ be an arbitrary vertex of $G$.  There exists a bias strategy
for the inverse-degree-biased walk such that the return time to $v$ in
the walk is at most
\[
\frac{\degree{v} + \sum_{x \neq v} \ws(x, v)\degree{x}}{\degree{v}}.
\]
\end{corollary}
\begin{proof}
The result follows from Lemma~\ref{lem:inverse}, setting $S = \{v\}$
and noting that the return time to $v$ equals the inverse of the
stationary probability at $v$.
\end{proof}

For our analysis, it is more convenient to work with an upper bound on
$\ws(x,v)$ (for each $x$ and $v$) as given the following lemma.  Let
$p(y,x)$ denote the weight of the shortest path from $y$ to $x$ in
$G$, where we assign a weight of $1/\degree{z}$ for each vertex $z$.
\begin{lemma}
\label{lem:ws}
For all vertices $x$ and $v$, $\ws(x,v)$ is at most $e^{-p(x,v)}$.
\end{lemma}
\begin{proof}
Let $P'$ denote a path between $x$ and $v$ such that $\ws(x,v)
= \sigma(P)$.  Then, we have
\[
\ws(x,v) = \prod_{y \in P'} (1 - 1/\degree{y}) \le \prod_{y \in P'} e^{-1/\degree{y}} = e^{-p(x,v)}.
\]
\end{proof}

\junk{
\begin{lemma}
\label{lem:biased_walk}
For any vertex $x$, we have
\[
R_x(x) \le \frac{\degree{x} + \sum_{y \neq x} (\degree{y} e^{-p(y,x)})}{\degree{x}},
\]
where $p(y,x)$ is the weight of the shortest path from $y$ to $x$ in
$G$, where we assign a weight of $1/\degree{z}$ for each vertex $z$.
\end{lemma}
\begin{proof}
Invoking Theorem~\ref{thm:azar} (with $S = \{x\}$), we obtain that there exists a 
\end{proof}

The proof of Lemma~\ref{lem:biased_walk} makes use of Theorem 2
in~\cite{BiasedRandomWalks}.
}

\begin{lemma}
\label{lem:bound}
For start vertex $u$ and target $v$, let $P$ be a shortest-hop path
from $u$ to $v$, given by the sequence of vertices $u = u_0, u_1, u_2,
\ldots, u_{|P|} = v$.  Then, we have
\[
\sum_{i = 0}^{|P| - 1} \degree{u_i} R(u_i) = O(n^{11/4}).
\]
\end{lemma}
\begin{proof}
By Lemma~\ref{lem:dominate} and Corollary~\ref{cor:inverse}, we have
\begin{eqnarray*}
\sum_{i = 0}^{|P| - 1} \degree{u_i} R(u_i) & \le  & \sum_{i = 0}^{|P| - 1} \degree{u_i} R(u_i)\\
& \le & \sum_{i = 0}^{|P| - 1} \degree{u_i}\left(\frac{\degree{u_i} + \sum_{x \neq u} \ws(x,u)\degree{x}}{\degree{u_i}}\right)\\
& \le & \sum_{i = 0}^{|P| - 1} \left(\degree{u_i} + \sum_{x\neq u_i} \degree{x} e^{-p(x,u_i)}\right)\\
& \le & 3n + \sum_{i=1}^{|P|-1} \sum_{x \neq u_i} \degree{x} e^{-p(x,u_i)}\\
& \le & 3n + \sum_{x \in V} \left(\degree{x} \sum_{i=1}^{|P|-1} e^{-p(x,u_i)}\right).
\end{eqnarray*}
(The second step follows from Corollary~\ref{cor:inverse}.  The third
step follows from Lemma~\ref{lem:ws}.  The fourth step follows from
the elementary claim that the sum of the degrees of vertices along any
shortest path in an $n$-vertex graph is at most $3n$; e.g.,
see~\cite[Proof of Theorem~2.1]{feige-rumor}.  The final step is
obtained by rearranging the summations.)

We next place an upper bound on $\sum_{x \in
V} \left(\degree{x} \sum_{i=1}^{|P|-1} e^{-p(x,u_i)}\right)$.  We
claim that the number of vertices $u_i$ in $P$ to which any vertex $x$
has paths of length at most $L$ is at most $2L$.  This is because two
vertices $a$ and $b$ that are more than $2L$ hops away in $P$ would
have a path of $2L$ hops via $x$, a contradiction.

Consider a vertex $u_i$, for some $i$.  Consider any path from $x$ to
$u_i$ that minimizes $p(x,u_i)$. No vertex in the graph has edges to
more than $\sqrt{n}$ vertices along this path.  If it did, then we
could short-cut and decrease the sum of the reciprocals of the degrees
$1/\sqrt{n} < \sqrt{n}/(n-1)$.  Therefore, any path that minimizes
$p(x,u_i)$ has a total degree of at most $n^{3/2}$.  If the length $L$
of the path is less than $\sqrt{n}$, then $p(x,u_i) \geq 0$; otherwise,
$p(x,u_i) \geq L^2/n^{3/2}$.

We now show that for any vertex $x$, $\sum_{i=1}^{|P|-1}
e^{-p(x,u_i)}$ is $O(n^{3/4})$.  We partition the $u_i$ according to
their hop-distance from $x$.  Fix $j$ and consider all vertices in $P$
that are distance between $2^{j-1}$ and $2^j$ from $x$.  By our claim
above, there are at most $2^j$ such vertices.  If $2^j \leq \sqrt{n}$,
then we use the trivial bound of $1$; otherwise, we use the bound
$2^{2j}/n^{3/2}$.  So we have
\begin{eqnarray*}
\sum_{u_i \in P} e^{-p(x,u_i)} & = & \sum_{j:  2^j \leq \sqrt{n}} O(2^j) + \sum_{j: 2^j > \sqrt{n}} \left(2^j *
  e^{-2^{2j}/n\ ^{3/2}}\right) \\
& = & O(n^{1/2}) + O(n^{3/4}) = O(n^{3/4}).
\end{eqnarray*}
Since $\degree{x}$ for each $x$ is at most $n-1$ and the length of $P$
is at most $n$, we get the desired bound of $O(n^{11/4})$.
\end{proof}

\begin{theorem}
\label{thm:cobra.general}
The hitting and cover times for a 2-cobra walk on a graph with $n$ vertices are $O(n^{11/4})$ and
$O(n^{11/4} \log n)$, respectively.  
\end{theorem}
\begin{proof}
For any start vertex $u$ and target vertex $v$, let $P$ be a
shortest-hop path from $u$ to $v$, given by the sequence of vertices
$u = u_0, u_1, u_2, \ldots, u_{|P|} = v$.  By
Lemma~\ref{lem:dominate} and Equation~\ref{eqn:hit} we have
\[
H(u,v) \le H^*(u,v) \le \sum_{i = 0}^{|P| - 1} \degree{u_i} R(u_i),
\] 
which, by Lemma~\ref{lem:bound}, is $O(n^{11/4})$.  The bound on the
cover time follows from Matthews' Theorem, Theorem~\ref{matthewscobra}.
\end{proof}

\section{Conclusion}

We have derived several improved results for cobra walks.  However,
there remain many open problems, and corresponding gaps in our
understanding.  We have shown that for $d$-dimensional grids on
$[0,n]^d$ the cover time is proportional to $n$, even for 2-cobra
walks.  The dependence on $d$, however, has not been determined, and more
generally we would like to find a larger collection of graph types for
which the cover time for 2-cobra walks is proportional to the
diameter.  Better results for general graphs are still open; we
optimistically conjecture that for 2-cobra walks the worst case cover
time on a graph with $n$ vertices is only $O(n \log n)$.  (The star
graph shows that it can be $\Omega(n \log n)$.)  Perhaps most
importantly, we believe there should be better methods available for
analyzing cobra walks.  Our results utilize techniques based on
parallelism and bias, but we do not believe they are as yet taking
full advantage of the power of cobra walks.

\vspace{-0.2in}

\newpage

\bibliographystyle{plain}
\bibliography{refs2,Distributed-RW2,expander,papers2,epidemics2}

\begin{thebibliography}{10}

\bibitem{AHKV03}
Micah Adler, Eran Halperin, Richard~M. Karp, and Vijay~V. Vazirani.
\newblock A stochastic process on the hypercube with applications to
  peer-to-peer networks.
\newblock In {\em Proceedings of the Thirty-fifth Annual ACM Symposium on
  Theory of Computing}, STOC '03, pages 575--584, New York, NY, USA, 2003. ACM.

\bibitem{AAKKLT}
Noga Alon, Chen Avin, Michal Koucky, Gady Kozma, Zvi Lotker, and Mark~R.
  Tuttle.
\newblock Many random walks are faster than one.
\newblock In {\em Proceedings of the Twentieth Annual Symposium on Parallelism
  in Algorithms and Architectures}, SPAA '08, pages 119--128, New York, NY,
  USA, 2008. ACM.

\bibitem{arthreya2005branching}
Siva~R. Arthreya and Jan~M Swart.
\newblock Branching-coalescing particle systems.
\newblock {\em Probability Theory and Related Fields}, 131(3):376--414, 2005.

\bibitem{7354403}
J.~Augustine, G.~Pandurangan, P.~Robinson, S.~Roche, and E.~Upfal.
\newblock Enabling robust and efficient distributed computation in dynamic
  peer-to-peer networks.
\newblock In {\em Foundations of Computer Science (FOCS), 2015 IEEE 56th Annual
  Symposium on}, pages 350--369, Oct 2015.

\bibitem{azar}
Yossi Azar, Andrei~Z Broder, Anna~R Karlin, Nathan Linial, and Steven Phillips.
\newblock Biased random walks.
\newblock In {\em Proceedings of the twenty-fourth annual ACM symposium on
  Theory of computing}, pages 1--9. ACM, 1992.

\bibitem{benjamini2010trace}
Itai Benjamini and Sebastian M{\"u}ller.
\newblock On the trace of branching random walks.
\newblock {\em arXiv preprint arXiv:1002.2781}, 2010.

\bibitem{berenbrink}
Petra Berenbrink, Colin Cooper, Robert Els\"{a}sser, Tomasz Radzik, and Thomas
  Sauerwald.
\newblock Speeding up random walks with neighborhood exploration.
\newblock In {\em Proceedings of the Twenty-first Annual ACM-SIAM Symposium on
  Discrete Algorithms}, SODA '10, pages 1422--1435, Philadelphia, PA, USA,
  2010. Society for Industrial and Applied Mathematics.

\bibitem{broder}
Andrei Broder.
\newblock Generating random spanning trees.
\newblock In {\em Foundations of Computer Science, 1989., 30th Annual Symposium
  on}, pages 442--447, Oct 1989.

\bibitem{sicomp}
Jen-Yeu Chen and Gopal Pandurangan.
\newblock Almost-optimal gossip-based aggregate computation.
\newblock {\em SIAM J. Comput.}, 41(3):455--483, 2012.

\bibitem{Chung05laplaciansand}
F.~Chung.
\newblock Laplacians and the cheeger inequality for directed graphs.
\newblock {\em Annals of Combinatorics}, 9:1--19, 2005.

\bibitem{cooper2012coalescing}
Colin Cooper, Robert Els\"{a}sser, Hirotaka Ono, and Tomasz Radzik.
\newblock Coalescing random walks and voting on graphs.
\newblock In {\em Proceedings of the 2012 ACM Symposium on Principles of
  Distributed Computing}, PODC '12, pages 47--56, New York, NY, USA, 2012. ACM.

\bibitem{DP05}
Nedialko~B. Dimitrov and C.~Greg Plaxton.
\newblock Optimal cover time for a graph-based coupon collector process.
\newblock In {\em Proceedings of the 32Nd International Conference on Automata,
  Languages and Programming}, ICALP'05, pages 702--716, Berlin, Heidelberg,
  2005. Springer-Verlag.

\bibitem{Dutta:2015:CRW:2821462.2817830}
Chinmoy Dutta, Gopal Pandurangan, Rajmohan Rajaraman, and Scott Roche.
\newblock Coalescing-branching random walks on graphs.
\newblock {\em ACM Trans. Parallel Comput.}, 2(3):20:1--20:29, November 2015.

\bibitem{ElsasserS09}
Robert Elsasser and Thomas Sauerwald.
\newblock Tight bounds for the cover time of multiple random walks.
\newblock volume 412, pages 2623 -- 2641, 2011.
\newblock Selected Papers from 36th International Colloquium on Automata,
  Languages and Programming (ICALP 2009).

\bibitem{feige1}
U.~Feige.
\newblock A tight upper bound on the cover time for random walks on graphs.
\newblock 1993.

\bibitem{feige2}
Uriel Feige.
\newblock A tight lower bound on the cover time for random walks on graphs.
\newblock {\em Random Struct. Algorithms}, 6(4):433--438, July 1995.

\bibitem{feige-rumor}
Uriel Feige, David Peleg, Prabhakar Raghavan, and Eli Upfal.
\newblock Randomized broadcast in networks.
\newblock {\em Random Structures \& Algorithms}, pages 447--460, 1990.

\bibitem{GANESH}
Ayalvadi~J. Ganesh, Laurent Massouli, and Donald~F. Towsley.
\newblock The effect of network topology on the spread of epidemics.
\newblock In {\em INFOCOM}, pages 1455--1466. IEEE, 2005.

\bibitem{MR0163361}
Theodore~E. Harris.
\newblock {\em The theory of branching processes}.
\newblock Die Grundlehren der Mathematischen Wissenschaften, Bd. 119.
  Springer-Verlag, Berlin, 1963.

\bibitem{KES}
David~A. Kessler.
\newblock Epidemic size in the sis model of endemic infections.
\newblock {\em Journal of Applied Probability}, 45(3):757--778, 09 2008.

\bibitem{LevinPeresWilmer2006}
David~Asher Levin, Yuval Peres, and Elizabeth~Lee Wilmer.
\newblock {\em Markov chains and mixing times}.
\newblock Providence, R.I. American Mathematical Society, 2009.
\newblock With a chapter on coupling from the past by James G. Propp and David
  B. Wilson.

\bibitem{Madras1992255}
Neal Madras and Rinaldo Schinazi.
\newblock Branching random walks on trees.
\newblock {\em Stochastic Processes and their Applications}, 42(2):255 -- 267,
  1992.

\bibitem{metropolis}
M.~Metropolis, A.~Rosenbluth, M.~Rosenbluth, A.~Teller, and M.~Teller.
\newblock Equation of state calculations by fast computing machines.
\newblock {\em Journal of Chemical Physics}, 21:1087--1092, 1953.

\bibitem{MU}
Michael Mitzenmacher and Eli Upfal.
\newblock {\em Probability and computing: Randomized algorithms and
  probabilistic analysis}.
\newblock Cambridge University Press, 2005.

\bibitem{Spielman}
Daniel Spielman.
\newblock Lecture notes in spectral graph theory, 2012.

\bibitem{sun2008brownian}
Rongfeng Sun and Jan~M Swart.
\newblock The brownian net.
\newblock {\em The Annals of Probability}, 36(3):1153--1208, 2008.

\bibitem{TCS-010}
Salil~P. Vadhan.
\newblock Pseudorandomness.
\newblock {\em Foundations and Trends in Theoretical Computer Science},
  7(1-3):1--336, 2011.

\bibitem{PIET}
Piet Van~Mieghem.
\newblock The n-intertwined sis epidemic network model.
\newblock {\em Computing}, 93(2-4):147--169, December 2011.

\end{thebibliography}

\end{document}